\documentclass[11pt]{article}

\usepackage[numbers,sort]{natbib}
\usepackage{fullpage}
\usepackage{amsmath,amssymb,amsthm}
\usepackage[all,ps,dvips]{xy}
\usepackage{mathrsfs}
\usepackage[titles]{tocloft}
\usepackage[normalem]{ulem}
\usepackage{algorithmic} 
\usepackage[linesnumbered,vlined,boxruled,boxed,algo2e]{algorithm2e}
\usepackage{url}
\usepackage[b]{esvect}
\usepackage{graphics, graphicx}

\usepackage{enumerate}

\usepackage{xcolor}
\definecolor{monvert}{rgb}{0.05,0.64,0.05}
\usepackage[pdftex,                %
    bookmarks         = true,
    bookmarksnumbered = true,
    pdfpagemode       = None,
    pdfstartview      = FitH,
    pdfpagelayout     = SinglePage,
    colorlinks        = true,
   citecolor         =blue,
    urlcolor          = magenta,
    pdfborder         = {0 0 0}
    ]{hyperref}

\def\C {\ensuremath{\mathbf{C}}}
\def\CC {\ensuremath{{\mathsf{C}}}}

\def\QQ {\ensuremath{\mathsf{Q}}}

\def\R {\ensuremath{\mathbf{R}}}

\def\Z {\ensuremath{\mathbb{Z}}}
\def\ZZ {\ensuremath{\mathsf{Z}}}
\def\RR {\ensuremath{\mathsf{R}}}
\def\K {\ensuremath{\mathbb{K}}}

\def\mA{\ensuremath{\mathbf{A}}}

\def\GL{\ensuremath{{\rm GL}}}

\def\crit{\ensuremath{{\rm crit}}}

\def\ext{\ensuremath{{\rm ext}}}
\def\Ext{\ensuremath{{\rm Ext}}}

\def\y {\ensuremath{\mathbf{y}}}

\def\z {\ensuremath{\mathbf{z}}}
\def\x {\ensuremath{\mathbf{x}}}

\def\X {\ensuremath{\mathbf{X}}}
\def\Y {\ensuremath{\mathbf{Y}}}

\def\limeps {\lim_{\varepsilon\rightarrow 0}}
\def\eps {\varepsilon}

\def\Veps {V_{\varepsilon}}

\def\Re {\ensuremath{ \R\langle \varepsilon \rangle}}
\def\Ce {\ensuremath{ \C\langle \varepsilon \rangle}}
\def\Ke {\ensuremath{ \K\langle \varepsilon \rangle}}
\def\RRe {\ensuremath{ \RR\langle \varepsilon \rangle}}
\def\CCe {\ensuremath{ \CC\langle \varepsilon \rangle}}

\def\RRz {\ensuremath{ \RR\langle \zeta \rangle}}
\def\CCz {\ensuremath{ \CC\langle \zeta \rangle}}

\def\RRze {\ensuremath{ \RR\langle \zeta \rangle\langle \eps\rangle}}
\def\CCze {\ensuremath{ \CC\langle \zeta \rangle\langle \eps\rangle}}

\def\RRzh {\ensuremath{ \RR\langle \zeta \rangle\langle \eta\rangle}}

\def\RRzhe {\ensuremath{ \RR\langle \zeta \rangle\langle \eta\rangle\langle \eps\rangle}}
\def\CCzhe {\ensuremath{ \CC\langle \zeta \rangle\langle \eta \rangle\langle \eps\rangle}}

\def\union {\ensuremath{ \bigcup}}
\def\intersection {\ensuremath{ \bigcap}}
\def\inter {\ensuremath{ \intersection}}

\newcommand{\be}{\begin{equation}}
\newcommand{\ee}{\end{equation}}

\def\ifm#1#2{\relax\ifmmode#1\else#2\fi} 

\def\vs{\smallskip}

\def\func#1{{\sf #1}}

\usepackage{etex,etoolbox}

\makeatletter
\providecommand{\@fourthoffour}[4]{#4}
\def\fixstatement#1{%
  \AtEndEnvironment{#1}{%
    \xdef\pat@label{\expandafter\expandafter\expandafter
      \@fourthoffour\csname#1\endcsname\space\@currentlabel}}}

\globtoksblk\prooftoks{1000}
\newcounter{proofcount}

\long\def\proofatend#1\endproofatend{%
  \edef\next{\noexpand\begin{proof}[{\noexpand \textbf{Proof of \pat@label}} ] }%
  \toks\numexpr\prooftoks+\value{proofcount}\relax=\expandafter{\next#1\end{proof} {\noexpand \newline}  {\noexpand \newline}}
  \stepcounter{proofcount}}

\def\printproofs{%
  \count@=\z@
  \loop
    \the\toks\numexpr\prooftoks+\count@\relax
    \ifnum\count@<\value{proofcount}%
    \advance\count@\@ne
  \repeat}
\makeatother

\title{A probabilistic algorithm to compute the real dimension of a semi-algebraic set}

\author{Mohab Safey El Din\\
  Universit\'e Pierre and Marie Curie and Institut Universitaire de France \\INRIA Paris-Rocquencourt\\ \texttt{Mohab.Safey@lip6.fr}\\ \\
  Elias Tsigaridas\\ 
  Universit\'e Pierre and Marie Curie  and INRIA Paris-Rocquencourt
  \\ \texttt{Elias.Tsigaridas@inria.fr}}

\date{April 6, 2013}

\newtheorem{definition}{Definition}
\newtheorem{theorem}{Theorem}
\newtheorem{corollary}[theorem]{Corollary}
\newtheorem{proposition}[theorem]{Proposition}
\newtheorem{lemma}[theorem]{Lemma}
\newtheorem{remark}[theorem]{Remark}

\newtheorem*{claim*}{Claim}

\fixstatement{theorem}
\fixstatement{corollary}
\fixstatement{proposition}
\fixstatement{lemma}

\begin{document}

\maketitle 

\begin{abstract}
  Let $\RR$ be a real closed field (e.g. the field of real numbers)
  and $\mathscr{S} \subset \RR^n$ be a semi-algebraic set defined as
  the set of points in $\RR^n$ satisfying a system of $s$ equalities
  and inequalities of multivariate polynomials in $n$ variables, of
  degree at most $D$, with coefficients in an ordered ring $\ZZ$
  contained in $\RR$.

  We consider the problem of computing the {\em real dimension}, $d$,
  of $\mathscr{S}$.  The real dimension is the first topological
  invariant of interest; it measures the number of degrees of
  freedom available to move in the set.  Thus, computing the real
  dimension is one of the most important and fundamental problems in
  computational real algebraic geometry.
  
  The problem is ${\rm NP}_{\mathbb{R}}$-complete in the
  Blum-Shub-Smale model of computation.  The current algorithms
  (probabilistic or deterministic) for computing the real dimension
  have complexity $(s \, D)^{O(d(n-d))}$, that becomes 
  $(s \, D)^{O(n^2)}$ in the worst-case.

  The existence of a probabilistic or deterministic algorithm for
  computing the real dimension with single exponential complexity with
  a factor better than ${O(n^2)}$ in the exponent in the worst-case,
  is a longstanding open problem.

  We provide a positive answer to this problem by introducing a
  probabilistic algorithm for computing the real dimension of a
  semi-algebraic set with complexity $( s\, D)^{O(n)}$.
\end{abstract}

\newpage

\section{Introduction}

A {\em semi-algebraic set} $\mathscr{S} \subset \RR^n$, where $\RR$ is
a real closed field, is defined as the set of points in $\RR^n$
satisfying a Boolean formula whose atoms are polynomial equalities and
inequalities.

Computational real algebraic or semi-algebraic geometry is the study
of effective algorithms for computing with semi-algebraic sets.
Besides being a fascinating and important research area on its own, it
is also one of the cornerstones of theoretical computer science.

Many important results rely on the foundations of real algebraic
geometry.  Let us mention non-linear computational geometry
\cite{ss-piano-I,ss-pian-II,CannyThese}, the recent breakthroughs in
combinatorial geometry on the discrete version of Kakeya problem
\cite{kms-dcg-2011,gk-am-2010,st-dcg-2012}, and the new algorithms for
non-negative matrix factorization \cite{arora2012, moitra2012} based
on testing the emptiness of semi-algebraic sets.
Last but not least, we emphasize the intrinsic connection between
computational real algebraic geometry and game theory, especially
stochastic games
\cite{bz-econ-1994,hklmt-stoc-2011,n-real-2003,sfv-siam-1997,cmh-ijgt-2008}.

Typical computational challenges in real algebraic geometry are 
algorithms for deciding the emptiness and/or computing at least one
point at each semi-algebraically connected component of a
semi-algebraic set
\cite{gh-dim-1991,RRS,BaPoRo06,r-fot-jsc-I-1992,bpr-jacm-1996,gv-subexp-jsc-1988},
algorithms to perform geometric operations such as projection (this
operation is tightly coupled with quantifier elimination)
\cite{BaPoRo06,g-tarski-jsc-1988,r-fot-jsc-I-1992,bpr-jacm-1996},
answering connectivity queries (roadmaps) \cite{BaPoRo96,HRSRoadmap,CannyThese},
computing the real dimension of a semi-algebraic set
\cite{v-ld-jsc-1999,k-focs-1997,k-rd-joc-1999} or computing more
sophisticated topological information, such as the number of
semi-algebraically connected components, the
Euler-Poincar\'e characteristic, Betti numbers
\cite{b-jsc-2006,bpr-stoc-2005,b-stoc-1996,BaPoRo06}.





Denote by $s$ the number of polynomials involved in the
description of a semi-algebraic set, by $n$ the number of 
variables, and by $D$ the maximum of the degrees of these polynomials.
We can solve almost all the problems in computational real algebraic
geometry using the generic approach of cylindrical algebraic
decomposition \cite{c-qe-1975} albeit in double exponential time,
$(s\,D)^{2^{O(n)}}$.
Even though huge effort has been invested the
last 25 years to derive algorithms with single exponential complexity
w.r.t.~the number of variables, there are problems that are still
missing an algorithm with such a complexity bound.  Moreover, even in
the case where the complexity is single exponential, the exponent is
not always $O(n)$.

Let us emphasize that improving the exponents in the complexity bounds
of algorithms in computational real algebraic geometry is not only a
theoretical challenge.  It introduces new algebraic and geometric
techniques that find applications in more general domains, and
eventually leads to efficient implementations for real-world problems.
For instance, the first improvement of the long-standing
$O(n^2)$ exponent in the complexity bound of Canny's probabilistic
algorithm \cite{CannyThese} to $O(n^{3/2})$ is based on a new
geometric connectivity result that introduced the use of a baby
steps giant steps algorithmic technique in this problem
\cite{SaSc11}.

On the other hand, the problem of computing the real dimension lacks,
up to now, an algorithm with single exponential complexity and
exponent $O(n)$.

\paragraph*{Problem statement and state-of-the-art.}
In this paper we address the problem of computing the real dimension
of a semi-algebraic set.
The following definition is in order:
\begin{definition}{\rm \cite[Section 5.3]{BaPoRo06}}
  \label{def:dim}
  Let $\mathscr{S}$ be a semi-algebraic-set of $\RR^n$, where $\RR$ is a
  real closed field.  The {\bf real dimension} of $S$ is
  the largest integer $d$ such that there exists an injective
  semi-algebraic map from $(0,1)^d$ to $\mathscr{S}$.  By definition the
  dimension of the empty set is $-1$.
\end{definition}

The best known complexity bound, in the worst case, for computing the
real dimension of a semi-algebraic set is due to Koiran and it is
$(s\,D)^{O(n^2)}$ \cite{k-rd-joc-1999}, where $s$ is the number of
polynomials used to describe the semi-algebraic set. It is based on
quantifier elimination techniques, see \cite[Alg.~14.10]{BaPoRo06} and
references therein.  A partial improvement of the $O(n^2)$ in the
exponent is due to Vorobjov \cite{v-ld-jsc-1999}. He presented an
algorithm with complexity $(s\,D)^{O(d(n-d))}$, where $d$ is the real
dimension. This bound is output sensitive, and when $d$ is a constant,
then it becomes $(s D)^{O(n)}$.  Basu, Pollack, and Roy
\cite{bpr-dsas-jms-2006} slightly improved the result of Vorobjov,
based on \cite{cfgk-focm-2003}.  They presented a complexity bound
that depends on whether $d\geq n/2$ or $d<n/2$.  This result has a
better dependence on the number of polynomials, $s$, than the one of
Vorobjov \cite{v-ld-jsc-1999}.

On the other hand, it is well understood that we can compute the
(Krull) dimension of an algebraic variety over algebraically closed
fields in time $D^{O(n)}$ \cite{gh-dim-1991,c-dim-jsc-1996}, see also
\cite{k-focs-1997}.  In the algebraically closed field case, we can
consider a sufficiently generic, random, collection of $d$
hyperplanes, for $1 \leq d \leq n$, and check whether their
intersection with the algebraic set under consideration is finite.
The largest $d$ where this is achieved imposes that the Krull
dimension of the algebraic set is $d$.
However, when we are interested in computing the real dimension of a
real algebraic, or semi-algebraic set, this strategy is not
applicable. 

It is of great interest to know if the problem of computing the dimension 
admits the same complexity bound in the real case and in the algebraically closed case. 
Quoting Koiran \cite{k-rd-joc-1999} {\it ``The main open problem is
  whether $\func{DIM}_{\mathbb{R}}$ [the real dimension] can be solved in
  time $(s\,D)^{O(n)}$''}.  Vorobjov \cite{v-ld-jsc-1999} also
mentions that {\it ``For a real variety $V$ existence of a
  probabilistic dimension algorithm with complexity bound $(s\,
  D)^{O(n)}$ is an open problem''}.

The purpose of the present work is to provide a positive answer to
this open problem.

Besides the intrinsic mathematical interest for an improved algorithm
for computing the real dimension, improvement of the complexity bound
has important consequences. Some algorithms in computational real
algebraic geometry consider the real dimension of a semi-algebraic set
as part of their input, e.g.~\cite[Theorem.~13.37]{BaPoRo06}.  Let us
also mention the recent bounds in
\cite{bb-arxiv-II-2013,bb-dcg-I-2012} on the number of
semi-algebraically connected components of a semi-algebraic set that
depend on the real dimension of some real algebraic set.  Moreover,
effective algorithms for computing the real dimension are needed to
estimate efficiently the degrees of freedom in robotic mechanisms (see
e.g. \cite{raghavan1993inverse,jin2002overconstraint}).

The problem is also very important from the complexity
theory point of view, as Koiran \cite{k-rd-joc-1999} proved that it is
${\rm NP}_{\mathbb{R}}$-complete in the Blum-Shub-Smale computation
model \cite{bss-model-1988}.

\paragraph*{Our results.}
We present an efficient algorithm for computing the real dimension of a semi-algebraic set.

Our algorithm reduces the unbounded case to the bounded one, using a
standard technique of computational real algebraic geometry introduced
in \cite{bpr-jacm-1996}.

Previous approaches for computing the real dimension of a
semi-algebraic set $\mathscr{S}$ rely on finding the largest integer
$d$ for which there exists a $d$-dimensional linear subspace, such that
the projection of $\mathscr{S}$ on this subspace has dimension $d$.
To do so they use quantifier elimination.  Therefore, the exponent in
the complexity is the dimension, $d$, multiplied by the number of
quantified variables, $n-d$.

The algorithm that we present, instead of projecting the
semi-algebraic set under consideration, exploits geometric
properties of fundamental objects of algebraic geometry, that is {\em
  polar varieties}.  Roughly speaking, polar varieties are the critical loci of
projections, e.g. \cite{BaGiHeMb97,BGHSS,BaGiHePa05} and references therein.
More precisely, we are able to prove the following:
Let $V$ be algebraic set defined by the polynomial equalities of the input
and $U$ the  open semi-algebraic set defined by the
polynomial inequalities of the input.
Then, up to a generic change of coordinates, $d$ is the largest
integer such that $\mathscr{S}$ is equal to the intersection of $U$
and the limit of the critical locus of the $d+1$ polar variety $V$,
after we perturbed it symbolically.  This way the computation of the
real dimension reduces to finding the largest integer $d$ with this
property.

The algorithm is probabilistic since we perform a random linear change
of coordinates in the beginning.  If we work over the integers, then
we denote by $\tau$ the maximum bit size of the coefficients {\em
  after} the linear change of coordinates.


Our main result is encapsulated in the following theorem:

\begin{theorem}\label{thm:main}
  Let $F = (f_1, \dots, f_p) \in \ZZ[X_1, \dots, X_n]^p$, $G = (g_1,
  \dots, g_s) \in \ZZ[X_1, \dots, X_n]^s$, where the degree of each
  $f_i$, resp. $g_j$, is at most $D$.  There exists a probabilistic
  algorithm for computing the real dimension of the semi-algebraic set
  defined by 
  \[
  f_1=\cdots=f_p=0, \quad g_1>0, \dots, g_s > 0 
  \enspace ,
  \]
  in  $(s \, D)^{O(n)}$ operations in $\ZZ$.

  If $\ZZ=\mathbb{Z}$,
  then the Boolean complexity of the algorithm $\tau \,(s \, D)^{O(n)}$.
\end{theorem}

To the best of our knowledge this is the first algorithm for computing
the real dimension of a real algebraic or semi-algebraic set within
this complexity bound.

\paragraph*{Organization of the paper.}
The rest of the paper is structured as follows: In the next Section we
present the necessary preliminaries from real algebraic geometry.  In
Section~\ref{sec:generic} we present the genericity properties that
the semi-algebraic set under consideration should satisfy. We prove
that a semi-algebraic set can satisfy these properties if we apply a
random linear change of coordinates.  Section~\ref{sec:geom-stmt}
presents the geometric result that is the crux of the matter of
our algorithm.  Finally, in Section~\ref{sec:Algo} we present the
algorithm \func{ComputeRealDimension} for computing the real
dimension, its various subroutines, the proof of correctness and the
complexity analysis.

\section{Preliminaries}
\label{sec:prelim}

In this section, we introduce some basic notions and some notations
that are used throughout the paper. 

As sketched in the introduction we will introduce some infinitesimals
to deform real algebraic sets. This will lead us to consider various
ground fields and semi-algebraic sets defined over these fields. We
refer the reader to \cite[Chapter 2]{BaPoRo06} for a more detailed
exposition of these notions on real fields, real closed fields,
infinitesimals and semi-algebraic sets. We will also use basic notions
coming from algebraic geometry since we will use the knowledge of the
dimension of some algebraic sets to deduce the real dimension of
semi-algebraic sets under study. For a more detailed exposition of
these notions, we refer the reader to \cite[Chapter 1]{Shafarevich77}.
The section finishes with some notions on critical points and polar
varieties that are extensively used in the sequel.

\paragraph*{Ground fields.}
Let $\QQ$ be a real field, $\RR$ be a real closed field and $\CC$ be
the algebraic closure of $\RR$. We consider a field $\K$ containing
$\QQ$ (e.g. $\RR$ or $\CC$) and let $\varepsilon$ be an infinitesimal.
In the sequel, $\K\langle\varepsilon\rangle$ stands for the Puiseux
series field.
We say that $z=\sum_{i\geq i_0}a_i\varepsilon^{i/q}\in
\K\langle \varepsilon\rangle$ is {\it bounded over $\K$} 
if and only if $i_0\geq 0$. 
We say that $\z=(z_1, \ldots, z_n)\in \Ke^n$ 
is {\it bounded over $\K$} if each $z_i$ is bounded over~$\K$.  
Given a bounded element $z\in \Ke$, 
we denote by $\limeps z$ the number $a_0 \in \K$. 
Given a bounded element $\z\in \Ke^n$, 
we denote by $\limeps \z$
the point $(\limeps(z_1), \ldots,  \limeps(z_n))\in \K^n$.
Given a subset $A\subset \Ke^n$, we denote by $\limeps(A)$ the set
$\{\limeps(z)\mid z\in A\mbox { and } z \mbox{ is bounded}\}.$
Given a semi-algebraic (resp. constructible) set $A\subset \RR^n$
(resp. $A\subset \CC^n$) defined by a quantifier-free formula $\Phi$
with polynomials in $\RR[X_1, \ldots, X_n]$, we denote by $\ext(A,
\RR\langle \varepsilon\rangle)$ (resp. $\ext(A, \CC\langle
\varepsilon\rangle)$) the set of solutions of $\Phi$ in $\RR\langle
\varepsilon\rangle^n$ (resp. $\CC\langle \varepsilon\rangle^n$).

In the sequel, we will work with $n$-variate polynomials with
coefficients in $\QQ$, $\QQ[\zeta]$ and $\QQ[\eps, \zeta]$ or $\eps$
and $\zeta$ are infinitesimals with $0< \eps< \zeta$. Sign conditions
on finite families of polynomials with coefficients in $\QQ$ define
semi-algebraic sets in $\RR^n$, those with coefficients in
$\QQ[\zeta]$ (resp. $\QQ[\eps, \zeta]$) define semi-algebraic sets in
$\RR\langle\zeta\rangle^n$ (resp. $\RR\langle\zeta\rangle
\langle\eps\rangle^n$).


\paragraph*{Basic definitions on algebraic sets.}
Let $\bar{\K}$ stand for an algebraic closure of $\K$.  We consider
{\em algebraic sets} in ${\bar \K}^n$ defined by polynomial
equations with coefficients in $\K$. A Zariski open set is a set whose
complementary is an algebraic set. A {\em constructible set} in ${\bar
  \K}^n$ is the set of common solutions of a system of a finite number of
$n$-variate polynomial equations and inequalities with coefficients in
$\K$.

Let $V\subset {\bar \K}^n$ be an algebraic set defined by polynomial
equations in $\K[X_1, \ldots, X_n]$.

We will consider the dimension of $V$, referring to
its {\em Krull dimension} (see e.g. \cite{Eisenbud95}). This notion of
dimension coincides with other notions coming from differential or
algebraic geometry (see e.g. \cite[Part II]{Eisenbud95}). Roughly
speaking, it is the number of {\em generic} hyperplanes such that
their intersection with $V$ is a finite set of points. The Krull
dimension of a constructible set is the Krull dimension of its
Zariski closure. If $W$ is another algebraic set and $V\subset W$,
then the Krull dimension of $V$ is less than or equal to the Krull
dimension of $W$.

The algebraic set $V$ is said to be irreducible if it cannot be
decomposed as the union of two algebraic sets defined by polynomial
equations with coefficients in $\K$. If $V$ is not irreducible it can
be uniquely decomposed as a finite union of irreducible algebraic
sets; these sets are called the irreducible components of $V$.

When all the irreducible components of $V$ have the same Krull
dimension, we say that $V$ is {\em equidimensional}. The ideal
associated to $V$ is the set of polynomials with coefficients in $\K$
which vanish on $V$. There exists a finite family of polynomials which
generate it in $\K[X_1, \ldots, X_n]$; let us denote it by $f_1,
\ldots, f_p$.

Assume that $V$ is equidimensional of Krull dimension $d$. A point
$\x\in V$ is called {\em regular} (or {\em smooth}) if the Jacobian matrix
$$
\left [\begin{array}{ccc}
  \frac{\partial f_1}{\partial X_1} & \cdots &   \frac{\partial f_1}{\partial X_n}\\ 
\vdots & & \vdots\\
  \frac{\partial f_p}{\partial X_1} & \cdots &   \frac{\partial f_p}{\partial X_n}\\ 
\end{array}\right ]
$$
has rank $n-d$ at $\x$. The kernel of the above Jacobian matrix at
$\x$ is the tangent space to $V$ at $\x$; we denote it by $T_\x V$.
Points in $V$ that are not regular are said to be {\em singular}. An
algebraic set with no singular points is smooth.

\paragraph*{Semi-algebraic sets and algebraic sets.} 
Let $\CC$ be the algebraic closure of $\RR$, and ${S}\subset \RR^n$ be
a semi-algebraic set. The smallest algebraic set containing ${\cal S}$
is called the Zariski closure of ${S}$. It is well-known that the
Krull dimension of the Zariski closure of ${S}$ equals the real
dimension of ${S}$ (see e.g. \cite[Proposition 2.8.2]{BoCoRo98}). In
particular if $W$ is an algebraic set that contains ${S}$, one can
conclude that the real dimension of ${S}$ is less than or equal to the
Krull dimension of $W$.

Let ${S}\subset \RR^n$ and $S'\subset \RR^n$. Consider a
semi-algebraic map $\varphi: S\to S'$ and $\RR'$ be a real closed
field containing $\RR$. We will consider the extension of $\varphi$ to
$\RR'$, denoted by $\ext(\varphi, \RR')$, as the semi-algebraic
function $\ext(S, \RR')\to \ext(S', \RR')$ whose graph is the
extension of the graph of $\varphi $ to $S'$.

If $S\subset S'$, then by Definition \ref{def:dim} the
real dimension of $S$ is less than or equal to the real dimension of
$S'$.

If $\x\in S$, we say that $\x$ is a {\em smooth point}
of $S$ if it is smooth in the Zariski closure of $S$; a semi-algebraic
set is said to be smooth if it is contained in the set of regular
points of its Zariski closure (this is direct a consequence of
\cite[Definition 3.3.4]{BoCoRo98}).

We also consider canonical projections $\pi_i: (x_1, \ldots,
x_n)\mapsto (x_1, \ldots, x_i)$.

\paragraph*{Change of variables.} 
Let $\K$ be a field containing $\QQ$ and $\bar \K$ be its algebraic
closure. Consider $f\in \K[X_1, \ldots, X_n]$ and $V\subset
{\bar\K}^n$ be the set of roots of $f$ in $\bar\K$ and $\mA\in
\GL_n(\QQ)$. We denote by $f^\mA$ the polynomial $f(\mA\X)$ and by
$V^\mA\subset {\bar\K}^n$. In other words, $V^\mA$ is the image of $V$
by the map $\x\to \mA^{-1}\x$.

Similarly, we will also consider change of variables on
semi-algebraic sets. If $S$ is a semi-algebraic set in $\RR^n$ and
$\mA\in \GL_n(\QQ)$, then $S^\mA$ denotes the image of $S$ by the map
$\x\to \mA^{-1}\x$.

\paragraph*{Critical points and polar varieties.}
Let $f_1, \ldots, f_p$ be polynomials in $\K[X_1, \ldots, X_n]$ and
$V\subset \CC^n$ be the algebraic set defined by
$f_1=\cdots=f_p=0$. For $1\leq i\leq n-p$, we define the {\em polar variety}
$\crit(\pi_i, V)$ as the set of points in $V$ at which all $p$-minors
of the truncated Jacobian matrix
$$
\left [\begin{array}{ccc}
  \frac{\partial f_1}{\partial X_{i+1}} & \cdots &   \frac{\partial f_1}{\partial X_n}\\ 
\vdots & & \vdots\\
  \frac{\partial f_p}{\partial X_{i+1}} & \cdots &   \frac{\partial f_p}{\partial X_n}\\ 
\end{array}\right ]
$$
vanish. By convention for $i\geq n-p+1$, we set $\crit(\pi_i, V)=V$
and for $i=0$ we set $\crit(\pi_i, V)=\emptyset$.

We refer to \cite{BGHSS} for a detailed study of polar
varieties. Below, we recall some basic results that are used in the paper.

Given a polynomial $f\in \K[X_1, \ldots, X_n]$ and $V\subset
{\bar\K}^n$ defined by $f=0$, we denote by $\crit(\pi_i, V)\subset
{\bar\K}^n$ defined by
$$
f= \frac{\partial f}{\partial X_{i+1}}=\cdots=\frac{\partial f}{\partial X_{n}}=0
\enspace.
$$
Assuming $f$ is square-free, the set of singular points of $V$ is
defined by the vanishing of $F$ and all its partial derivatives. The
set $\crit(\pi_i, V)$ is called {\em polar variety associated to
$\pi_i$}. When $V$ is smooth, $\crit(\pi_i, V)$ is the set of {\em
  critical points} of the restriction of $\pi_i$ to $V$ (i.e. the set
of regular points $\x\in V$ such that $\pi_i(T_\x V)$ has dimension
$\leq i-1$).  When $V$ is not smooth $\crit(\pi_i, V)$ is the union of
the singular points of $V$ and the critical points of the restriction
of $\pi_i$ to $V$.

Let us also mention that for $i=1$, $\crit(\pi_1, V)$ contains the local
minimizers and maximizers of the restriction of $\pi_1$ to
$V\cap\RR^n$.


\section{Genericity properties}
\label{sec:generic}

In this paper, a property is called generic (in some suitable
parameter space) if it holds in a non-empty Zariski open subset of the
parameter space under consideration.

\begin{definition}\label{def:generic-position}
  Let $f\in \QQ(\zeta)[X_1, \ldots, X_n]$, $V\subset \CCz^n$ be the
  algebraic variety defined by $f=0$, and $V_\eps \subset
  \CCz\langle\eps\rangle^n$ be the algebraic variety defined by
  $f-\eps=0$.  We say that $f$ satisfies property $\sf N$ if the
  following conditions hold:
  
  \begin{enumerate}[${\sf N}_1$]
  \item for all $1\leq i \leq n$, $\crit(\pi_i, V_{\eps})$ is
    either empty or is smooth and equidimensional of Krull dimension
    $i-1$;
  \item for all $\x\in V\cap \RRz^n$, $\pi_d^{-1}(\pi_d(\x))\cap (V\cap
    \RRz^n)$ is finite where $d$ is larger than or equal to the
    real dimension of $V$ at $\x$.
  \end{enumerate}
\end{definition}

Note that in the above definition, we consider a polynomial with
coefficients in $\QQ(\zeta)$. We state below that for a generic choice
of an $n\times n$ invertible matrix $\mA$ with entries in $\QQ$,
$f^\mA$ satisfies $\sf N$.

Indeed we need such a statement because our algorithm performs
symbolic manipulations on the input by introducing an infinitesimal
$\zeta$ (to reduce the study to bounded semi-algebraic sets) and next
chooses randomly $\mA$ to ensure that after applying the change of
variables $\x\to \mA^{-1}\x$ some polynomial satisfies $\sf N$.

Thus, the rest of this Section is devoted to prove that up to a generic
choice of $\mA\in \GL_n(\QQ)$, $f^\mA$ satisfies $\sf N$.

\begin{proposition}\label{prop:normalisgeneric}
  There exists a non-empty Zariski open set $\mathscr{O}\subset
  \GL_n(\CC)$ such that for $\mA\in \mathscr{O}\cap \GL_n(\QQ)$,
  $f^\mA$ satisfies $\func{N}$.
\end{proposition}

\begin{proof}
  We will prove that there exists a non-empty Zariski open set
  $\mathscr{O}'\subset \GL_n(\CC)$ (resp.  $\mathscr{O}''\subset
  \GL_n(\CC)$) such that for $\mA\in \mathscr{O}'\cap\GL_n(\QQ)$,
  $f^\mA$ satisfies $\sf N_1$ (resp. $\sf N_2$). Taking
  $\mathscr{O}=\mathscr{O}'\cap\mathscr{O}''$ 
  is sufficient to conclude.

  We start with ${\sf N}_1$. \cite[Proposition 3]{BaGiHeMb97} (see
  also \cite[Theorem 6]{BGHSS} for a more general statement) states
  that when $f$ has coefficients in $\QQ$ and defines a smooth
  algebraic set $V\subset\CC^n$, there exists a non-empty Zariski open
  set $\Omega'\subset \GL_n(\CC)$ such that for all $\mA\in
  \GL_n(\QQ)\cap\Omega'$ and $1\leq i \leq n$, $\crit(\pi_i, V^\mA)$ is
  either empty or is equidimensional of Krull dimension $i-1$. The
  proof of this result is based on the use of the Weak Transversality
  Theorem of Thom-Sard (see e.g. \cite{golubitsky1973}). It allows to
  characterize the set of ``bad'' matrices, i.e. the complement of
  $\Omega'$ in $\GL_n(\CC)$ as the smallest algebraic set containing
  the critical values of a polynomial mapping with coefficients in the
  same base field as the one containing the coefficients of $f$.

  By \cite[Lemma 3.5]{RRS}, $V_\varepsilon$ is smooth. Thus, one can
  apply {\em mutatis mutandis} the proof of \cite[Proposition
  3]{BaGiHeMb97} to $f-\eps$ with $\CCze$ as a base field. We obtain
  the existence of a non-empty Zariski open set $\Omega'\subset
  \GL_n(\CCze)$ such that for $\mA\in \Omega'$ and $1\leq i \leq n$,
  $\crit(\pi_i, V_\varepsilon^\mA)$ is either empty or is
  equidimensional of Krull dimension $i-1$. Recall that
  $\GL_n(\CCze)-\Omega'$ is Zariski closed and is characterized as the
  set of critical values of a polynomial mapping in $\QQ(\eps, \zeta)$
  since $f-\eps$ has coefficients in $\QQ(\eps, \zeta)$.  If we
  multiply this polynomial by the least common multiple of the
  denominators of its coefficients, we obtain a polynomial $P$ with
  coefficients in $ \QQ[\eps, \zeta]$. Without loss of generality, we
  can assume that the coefficients of $P$ have no non-trivial gcd.
  Hence, $P$ can be written as
  \[
  P=P_0+\eps^{v_\eps} Q_\eps +\zeta^{v_\zeta} Q_\zeta+\eps\zeta Q
  \enspace ,
  \]
  where $v_\eps$ and $v_\zeta$ are positive integers and 
  \begin{itemize}
  \item $P_0$ has coefficients in $\QQ$ (it
    is obtained by substituting $\eps$ and $\zeta$ by $0$ in $P$);
  \item $Q_\eps$ (resp. $Q_\zeta$) has coefficients in $\QQ[\eps]$
    (resp. $\QQ[\zeta]$) and is not identically $0$ when $\eps$
    (resp. $\zeta$) is substituted to $0$;
  \item $Q$ has coefficients in $\QQ[\eps, \zeta]$. 
  \end{itemize}
  Note that since the coefficients of $P$ have no non-trivial gcd, at
  least one of the three polynomials $P_0, Q_\eps$ and $Q_\zeta$ are not
  identically $0$. If $P_0\neq 0$ (resp. $Q_\eps\neq 0$ or $Q_\zeta\neq
  0$), we define $\mathscr{O}'\subset \GL_n(\CC)$ as the non-empty
  Zariski defined by $P_0\neq 0$ (resp. $Q_\eps\neq 0$ or $Q_\zeta\neq
  0$).

  Now remark that since $\eps$ and $\zeta$ are infinitesimals, for all
  $\mathscr{O'}\cap \GL_n(\QQ)\subset \Omega'$; in other words for all
  $\mA\in \mathscr{O}'\cap\GL_n(\QQ)$, $f^\mA$ satisfies $\sf N_1$ as
  requested.
  
  \medskip Now we deal with $\sf N_2$. Below, by abuse of notations
  the extensions of cartesian products $]0, 1[^i$ in $\RRz$ are
  denoted by $]0, 1[$; also by convention $]0, 1[^0=\{0\}$.  The set
  $V \inter \RRz^n$ is semi-algebraic and so we can partition it into
  smooth semi-algebraic sets $\mathcal{S}_1, \ldots, \mathcal{S}_\ell$
  \cite[Chapter 5, Section 5]{BaPoRo06}, and homeomorphic to $]0,
  1[^{d_1}, \ldots, ]0, 1[^{d_\ell}$, respectively, where $0\leq k
  \leq n$. By \cite[pp. 47]{BCR}, for $\x\in V\inter\RRz^n$, the local
  real dimension of $V\inter \RRz^n$ at $\x$ is given by $\max _{\x\in
    \overline{\mathcal{S}_i}}d_i$, where $\overline{\mathcal{S}_i}$
  denotes the Euclidean closure of ${\mathcal{S}_i}$.

  For $1\leq i \leq \ell$, we denote by $V_i$ the Zariski closure of
  $\mathcal{S}_i$.  By \cite[Proposition 2.8.2]{BoCoRo98} note that
  the Krull dimension of $V_i$ is $d_i$, for $1\leq i \leq \ell$.  By
  Noether normalization \cite{am-book}, there exists a non-empty
  Zariski open set $\Omega''_i\subset \GL_n(\CCze)$ such that for all
  $\mA\in \Omega''_i$ and $\x\in \CCz^{d_i}$, $\pi_{d_i}^{-1}(\x)\cap
  V_i^{\mA}$ is finite. As a consequence, for all $\x\in \RRz^{d_i}$,
  $\pi_{d_i}^{-1}(\x)\cap \mathcal{S}_i$ is finite. We let
  $\Omega''=\inter_{i=1}^\ell \Omega''_i$.  The complement of
  $\Omega''$ in $\GL_n(\CCze)$ can be characterized using algebraic
  elimination algorithms (such as Gr\"obner bases) run with parameters
  as for entries of a generic matrix; we refer to \cite{l-cpnnl-1989}
  for such algebraic algorithms for detecting defects of Noether
  normalization. As above, since the input has coefficients in
  $\QQ(\zeta)$, $\Omega''$ can be defined by a polynomial inequality
  $P\neq 0$ with coefficients in $\QQ[\eps, \zeta]$: Following {\em
    mutatis mutandis} the approach for proving ${\sf N}_1$, one can
  deduce from the inequality $P\neq 0$ another inequality with
  coefficients in $\QQ$ defining a non-empty Zariski open set
  $\mathscr{O}''$ such that $\mathscr{O}''\cap\GL_n(\QQ)\subset
  \Omega''$. This finishes the proof.
\end{proof}


\section{Geometric Statement}
\label{sec:geom-stmt}

In this Section we let $f$ and $g_1, \ldots, g_s$ be polynomials in
$\ZZ[X_1, \ldots, X_n]$, $S\subset \RR^n$ be the semi-algebraic set defined by 
$$f=0, \quad g_1>0, \ldots, g_s>0$$ 
and $U\subset \RR^n$ be the open semi-algebraic set defined by $g_1>0,
\ldots, g_s>0$. We will also consider the algebraic set $V\subset
\CC^n$ defined by $f=0$.  Given an infinitesimal $\eps$, we denote by
$V_{\eps}\subset \CCe^n$ the algebraic set defined by $f=\eps$.

The rest of this Section is devoted to prove this result below.

\begin{proposition}\label{prop:criterion:dimension}
  Assume that $f$ satisfies $\func{N}$, is non-negative over $\RR^n$,
  and that $V\cap\RR^n$ is bounded.  Then for $0\leq i \leq n$,
  $\left (\limeps\crit(\pi_i, V_{\eps})\right ) \inter U=S$ if and
  only if the real dimension of $S$ is $\leq i-1$.
\end{proposition}

\begin{remark}
  \cite[Proposition 3.4]{bb-dcg-I-2012} provides a similar statement
  but with different assumptions on $f$ that are not suitable for
  our setting.
\end{remark}

Before proving the above result, we start with a few lemmata.

\subsection{Auxiliary results}
\label{sec:aux-lem}

The following proposition is a variant of statements that are commonly used
in computational real algebraic geometry (see
e.g. \cite{br-radii-2010, RRS} or \cite[Proposition 12.38]{BaPoRo06}).

We let $\zeta$ be an infinitesimal. 

\begin{proposition}\label{prop:comp-deform}
  Assume that $f$ is non-negative over $\RR^n$ and that $V\cap\RR^n$ is
  non-empty. Let $C$ be a semi-algebraically connected component of
  $V\cap\RR^n$.  Then there exist semi-algebraically connected
  components $C_{\eps,1}, \ldots, C_{\eps,\ell}$ of $V_{\eps} \inter
  \RRe^n$, such that $C=\union_{i=1}^\ell \limeps C_{\eps,i}$.

  Moreover if there exists a ball $B\subset \RR^n$ such that $C\subset
  B$ and $C$ does not intersect the boundary of $B$, then $C_{\eps,
    i}\subset \ext(B, \RRe)$ and $C_{\eps, i}$ does not intersect the
  boundary of $\ext(B, \RRe)$, for $1\leq i \leq \ell$.
\end{proposition}
\begin{proof}
  Consider $\x\in C$. By assumption, $f$ is non-negative over
  $\RR^n$. Since $f(\x)=0$ there exists a semi-algebraically connected
  component $\mathcal{S}\subset \RR^n$ of the semi-algebraic set defined by $f>0$
  such that $\x$ is in the closure of $\mathcal{S}$ (for the Euclidean
  topology). Then, for all $r>0$ the ball $B(\x, r)$ centered at $\x$
  of radius $r$ contains a point of $\mathcal{S}$ (at which $f$ is positive). By
  the curve selection Lemma \cite[Theorem 3.19]{BaPoRo06} there exists
  a continuous semi-algebraic function $\gamma: [0, 1]\to \mathcal{S}$ with
  $\gamma(0)=\x$ and $f(\gamma(t))>0$ for $t\neq 0$.

  Consider the extensions $\tilde \gamma=\ext(\gamma, \RRe)$ and
  $\tilde f=\ext(f, \RRe)$. By the semi-algebraic intermediate value
  Theorem \cite[Theorem 2.11]{BaPoRo06}, there exists $t_{\eps}\in
  \ext([0,1], \RRe)$ such that
  $\tilde f(\tilde \gamma(t_{\eps}))=\eps$. We denote
  $\tilde \gamma(t_{\eps})$ by $\x_\eps$ and we have
  $\limeps\x_\eps=\x$ because $\x_\eps\in \ext(B(\x, r), \RRe)$ for
  all $r>0$. Also, let $C_{\x_\eps}$ be the semi-algebraically
  connected component of $V_\eps\cap\RRe^n$ that contains $\x_\eps$;
  we associate to $\x$ this semi-algebraically connected component
  $C_{\x_\eps}$ of $V_\eps\cap\RRe^n$.

  Since there are finitely many semi-algebraically connected
  components of $V_\eps\cap\RRe^n$, there are finitely many
  semi-algebraically connected components of $C_{\eps, 1}, \ldots,
  C_{\eps, \ell}$ which are associated to a point $\x$ in $C$.
  This proves that  $C\subset \union_{i=1}^\ell \limeps C_{\eps,i}$.

  We prove now that $\union_{i=1}^\ell \limeps C_{\eps,i}\subset C$.
  Let $\x\in \limeps C_{\eps,i}$ for some $1\leq i \leq \ell$; we need
  to prove that $\x\in C$.  Then, there exists $\x_\eps\in C_{\eps,
    i}$ bounded over $\RR$ such that $\limeps\x_\eps=\x$. Now, recall
  that $C_{\eps, i}$ is associated to some point $\x'\in C$; this
  means that there exists $\x'_\eps\in C_{\eps, i}$ bounded over $\RR$
  such that $\limeps\x'_\eps=\x'$.  Since $C_{i, \eps}$ is
  semi-algebraically connected, there exists a continuous
  semi-algebraic function $\gamma:[0,1]\rightarrow C_{i,\eps}$ such
  that $\gamma(0)=\x_\eps$ and $\gamma(1)=\x'_\eps$; we have that
  $\Gamma=\gamma([0, 1])$ is a semi-algebraically connected
  semi-algebraic set. By \cite[Theorem 3.20]{BaPoRo06},
  $\Gamma=\gamma([0,1])$ is closed and bounded. By \cite[Proposition
  12.36]{BaPoRo06}, we conclude that $\limeps\Gamma$ is
  semi-algebraically connected. Since $\limeps$ is a ring
  homomorphism, $\limeps\Gamma$ is contained in $V\cap\RR^n$. Now
  notice that $\limeps(\gamma(0))=\x$ and that
  $\limeps(\gamma(1))=\x'$. Since we have proved that $\limeps\Gamma$
  is semi-algebraically connected we deduce that $\x\in C$ as
  requested.

  \medskip 

  Now, we assume that $C$ is bounded and let $B\subset \RR^n$ be a
  ball such that $C\subset B$ and $C$ does not intersect the boundary
  of $B$.  To conclude the proof it remains to prove that for $1\leq i
  \leq \ell$, $C_{\eps,i}\subset \ext(B, \RRe)$ and
  $C_{\eps,i}$ does not intersect the boundary of $\ext(B,\RRe)$.

  Consider $\x_\eps\in C_{\eps, i}$ bounded over $\RR$ such that
  $\limeps\x_\eps\in C$.  Such a point exists, as we argued in the
  first part of the proof that $\limeps C_{\eps, i}\subset C$.
  
  We argue by contradiction.  Assume that there exists $\x'_\eps$ in
  the Euclidean closure of $C_{\eps,i}-\ext(B, \RRe)$. Since $C_{\eps,
    i}$ is semi-algebraically connected, there exists a continuous
  semi-algebraic function $\gamma:[0, 1]\rightarrow C_{\eps, i}$ such
  that $\gamma(0)=\x_\eps$ and $\gamma(1)=\x'_\eps$. Note that by the
  intermediate value theorem \cite[Theorem 2.11]{BaPoRo06} (applied to
  the polynomial defining the boundary of $B$) there exists
  $\vartheta\in [0, 1]$ such that $\gamma(\vartheta)$ lies in the
  boundary of $\ext(B, \RRe)$. We deduce that
  $\limeps\gamma(\vartheta)$ belongs to the boundary of $B$.

  By \cite[Theorem 3.20]{BaPoRo06}, $\gamma([0,1])$ is closed and
  bounded. We also notice that $\gamma([0,1])$ is
  semi-algebrai\-cal\-ly connected. By \cite[Proposition
  12.36]{BaPoRo06}, we conclude that $\limeps\gamma([0,1])$ is
  semi-algebraically connected. We deduce that $\limeps(\gamma([0,
  1]))\subset C$. We deduce that $\limeps\gamma(\vartheta)$ belongs to
  $C$.

  Consequently, we have $\limeps\gamma(\vartheta)$ belongs to $C$ and
  to the boundary of $B$. This contradicts the fact that, by
  assumption, $C$ does not intersect the boundary of $B$.  Thus, we
  conclude that $C_{\eps, i}\subset \ext(B, \RRe)$ and $C_{\eps, i}$
  does not intersect the boundary of $\ext(B, \RRe)$.
\end{proof}

The following lemma relates the real dimension of a semi-algebraic set
in $\RRe^n$ with the real dimension of its image by $\limeps$. Its
proof uses quite standard tools from real algebraic geometry.

\begin{lemma}\label{lemma:dim:limites}
  Let $\mathcal{S}_{\eps}\subset \RRe^n$
  (resp. $\mathcal{Z}_\eps\subset \CCe^n$) be a semi-algebraic
  (resp. constructible) set and $d$ be its real (resp. Krull)
  dimension. Then $\limeps \mathcal{S}_{\eps}$ (resp. $\limeps
  \mathcal{Z}_\eps$) has real (resp. Krull) dimension less or equal to
  $d$.
\end{lemma}
\begin{proof}
  The proof of this statement relies on
  \cite[Proposition~5.29]{BaPoRo06} which states that if $A$ is a
  semi-algebraic subset of $\RR^m$ and $h: A \rightarrow \RR^{n}$ is a
  semi-algebraic mapping, then the real dimension of $h(A)$ is less
  than or equal to the real dimension of $A$.  Thus, it suffices to
  prove that the ring homomorphism $\limeps$ is a semi-algebraic
  function, i.e. a function whose graph is a semi-algebraic
  function. This is a quite routine statement in real algebraic
  geometry; we give the proof since we could not find a reference stating
  it explicitly. To do that, we reuse some ingredients of
  \cite[Proposition~12.36]{BaPoRo06}.

  Recall that by assumption $\mathcal{S}_{\eps}$ is a semi-algebraic set of
  $\RRe^n$ and that, by definition $\RRe$ is the real closure of
  $\RR(\eps)$. Then, by \cite[Proposition 2.82]{BaPoRo06}, there exists
  a quantifier-free Boolean formula of conjunctions and disjunctions
  of polynomials in $\RR[\eps][X_1, \ldots, X_n]$ which define
  $\mathcal{S}_\eps$. Below, we denote by $\Psi(\X, \eps)$ such a formula.

  Consider the following set 
  \begin{displaymath}
    T = \{ (\x, x_{n+1}) \in \RR^{n+1} \,|\, \Psi(\x, x_{n+1}) \wedge x_{n+1} >0\}
    \enspace,
  \end{displaymath}
  and let $\overline T$ its closure and $H$ be the hyperplane defined
  by $X_{n+1}=0$ in $\RR^{n+1}$.  As in the proof of
  \cite[Proposition~12.36]{BaPoRo06}, the following equalities hold:
  $$
  \mathcal{S} = \limeps \mathcal{S}_{\eps} = {\overline T}  \,\inter\, H. 
  $$
  We can  express the limit using a formula in the
  first order theory of reals.  For $\x \in \mathcal{S}_{\eps}$ and $\y \in \mathcal{S}$ we
  have
  \begin{displaymath}
    \Phi_1 := \left[
      (\forall r > 0) \,(\exists \eps_0) : (\forall \eps) \, (0 < \eps < \eps_0) 
      \Rightarrow \| \x - \y \|^2 < r^2 \right]  \Leftrightarrow    \limeps \x = \y
    \enspace .
  \end{displaymath}
  By quantifier elimination over the reals (see \cite{Tarski,c-qe-1975}
  or \cite[Theorem 2.77]{BaPoRo06}), there exists a quantifier free
  formula $\Psi_1$ which is equivalent to $\Phi_1$. 
  The set
  \begin{displaymath}
    \{ (\x, \y) \in \RRe^n \times \RR^{n} \,|\, \Psi \wedge (\y \in \mathcal{S})  \wedge \Psi_1\}
    \enspace,
  \end{displaymath}
  is the graph of $\limeps$ and is semi-algebraic.  This concludes the
  proof.

  The proof in the complex case uses exactly the same techniques as
  above transposed to algebraically closed fields, i.e. if $A$ and $B$
  are constructible sets and $h: A\rightarrow B$ is a regular map (see
  \cite[Chapter 1]{Shafarevich77}), then the Krull dimension of $h(A)$
  is less than or equal to $h(B)$ and quantifier elimination over
  algebraically closed fields (see e.g. \cite[Chapter 1]{BaPoRo06}).
\end{proof}

\begin{lemma}\label{lemma:locallyfinite}
  Consider $\x\in V\cap\RR^n$ and assume that $f$ is non-negative over
  $\RR^n$ and that there exists a neighbourhood $\mathscr{B}\subset
  \RR^n$ of $\x$ such that $V \inter \mathscr{B} \inter \RR^n$ is a
  finite set.
  Then, there exists $\x_\eps\in \crit(\pi_1, V_\varepsilon)$ such
  that $\limeps \x_\eps=\x$.
\end{lemma}

\begin{proof}
  By assumption $V \inter \mathscr{B} \inter \RR^n$ is a finite set.
  Hence, there exists a sufficiently small $r\in \RR$, $r >0$, such
  that $f$ is positive at all points in a ball, $B$, with center at
  $\x$ and radius $r$, except $\x$.  That is $\x$ is the only point of
  $V \inter R^n$ in $B$; hence $\{\x\}$ is a bounded semi-algebraically
  connected component of $V \inter \RR^n$.
  

  By Proposition~\ref{prop:comp-deform}, there exist
  semi-algebraically connected components $C_{\eps, 1}, \ldots,
  C_{\eps, \ell}$ of $V_\eps, \RRe^n$ such that
  $\{\x\}=\union_{i=1}^\ell \limeps C_{\eps,i}$,
  $C_{\eps, i}\subset \ext(B, \RRe)$,
  and  $C_{\eps, i}$ does not intersect the boundary 
  of $\ext(B, \RRe)$, $1\leq i \leq \ell$.

  We deduce that for $1\leq i \leq \ell$, 
  $C_{\eps, i}$ is closed and bounded. 
  However, $C_{\eps, i} \subset \ext(B, \RRe)$ and it is closed and
  bounded; as a consequence, it has a non-empty intersection with
  $\crit(\pi_1, V_\eps)$. Since we already observed that 
  $\limeps C_{\eps, i}\subset \{\x\}$, we deduce that for all $\x_\eps\in
  \crit(\pi_i, V_\eps)\cap C_{\eps, i}$, $\limeps \x_\eps=\x$. Our
  conclusion follows.
\end{proof}

\begin{lemma}\label{lemma:polarnonempty}
  Assume that $f$ is non-negative over $\RR^n$ and that $V\cap\RR^n$
  is bounded and non-empty. Then, for $1\leq i \leq n$, $\crit(\pi_i,
  V_\eps)$ is not empty and intersects all bounded semi-algebraically connected
  components of $V_\eps\cap\RRe^n$.
\end{lemma}

\begin{proof}
  Since $V\cap\RR^n$ is bounded and non-empty there exists a
  semi-algebraically connected component $C$ and a ball $B\subset
  \RR^n$ such that $C\subset B$ and $C\cap B=\emptyset$.  By
  Proposition~\ref{prop:comp-deform}, we deduce that there exist
  semi-algebraically connected components $C_{\eps, 1}, \ldots,
  C_{\eps, k}$ of $V_\eps\cap\RRe^n$ such that $C=\union_{\ell=1}^k
  \limeps C_{\eps,\ell}$. Moreover since we assume that $V\cap \RR^n$
  is bounded, there exists a ball $B\subset \RR^n$ such that $C\subset
  B$ and $C\cap B=\emptyset$. Using again
  Proposition~\ref{prop:comp-deform}, we conclude that $C_{\eps,
    \ell}\subset \ext(B, \RRe)$ for $1\leq \ell \leq k$; hence
  $C_{\eps, \ell}$ is closed and bounded.

  Below, we prove that any bounded semi-algebraically connected
  component $C_\eps$ of $V_\eps\cap \RRe^n$ has a non-empty
  intersection with $\crit(\pi_i, V_\eps)$.  Then $\pi_1(C_{\eps})$
  is closed and bounded (see \cite[Theorem 3.20]{BaPoRo06}). The
  extreme values of $\pi_1(C_{\eps})$ are attained at critical
  points of the restriction of $\pi_1$ to $V_\eps\cap \RRe^n$. In
  other words, $C_{\eps}$ has a non-empty intersection with
  $\crit(\pi_1, V_\eps)$. Now, note that $\crit(\pi_1, V_\eps)\subset
  \crit(\pi_i, V_\eps)$, by definition.
\end{proof}

\begin{lemma}\label{lemma:pointscritiquesetfibres}
  Let $f\in \CC[X_1, \ldots, X_n]$ and $V \subset \CC^n$ be defined by
  $f=0$.  Take $\alpha=(\alpha_1, \ldots, \alpha_{i-1}) \in \CC^i$ and
  let $V_{i,\alpha}$ the algebraic set 
  $V \inter \pi_{i-1}^{-1}(\alpha)$ and $\varphi_i$ be the canonical projection 
  $(x_1,  \ldots, x_n)\mapsto x_i$.  Then, the following  holds:
  $$\crit(\varphi_i, V_{i,\alpha})\subset\crit(\pi_i, V).$$
\end{lemma}

\begin{proof}
  By definition $\crit(\pi_{i}, V)$ is defined by the vanishing of $f$
  and of the partial derivatives $\frac{\partial f}{\partial X_{i+1}},
  \ldots, \frac{\partial f}{\partial X_n}$.  Besides, the set
  $V_{i,\alpha}$ is defined by the system
  \begin{equation*}
    \label{eq:pointscritiquesetfibres}
    f =0, X_1-\alpha_1=\cdots=X_{i-1}-\alpha_{i-1}=0.
  \end{equation*}
  By definition, $\crit(\varphi_i, V_{i,\alpha})$ is defined by the
  above equations and the vanishing of the maximal minors of the
  Jacobian matrix associated to $f,
  X_1-\alpha_1,\ldots,X_{i-1}-\alpha_{i-1}, X_i$. The triangular shape
  of this Jacobian matrix implies that $\crit(\varphi_i,
  V_{i,\alpha})$ is defined by $f =0,
  X_1-\alpha_1=\cdots=X_{i-1}-\alpha_{i-1}=0$ and the vanishing of the
  partial derivatives $\frac{\partial f}{\partial X_{i+1}}, \ldots,
  \frac{\partial f}{\partial X_n}$.

  The lemma follows from the observation that the system that defines
  $\crit(\pi_i, V_i)$ is contained in the system that defines
  $\crit(\varphi_i, V_{i,\alpha})$.
\end{proof}

\subsection{Proof of Proposition \ref{prop:criterion:dimension} and Consequences}

\begin{proof}[Proof of Proposition \ref{prop:criterion:dimension}.]
  First we prove the necessary condition. Assume that for a given $i\in \{0, \ldots, n\}$ it holds 
  $$
  \left (\limeps(\crit(\pi_i, V_\eps))\right )\inter U=S \enspace.
  $$ 

  We prove that this implies that the real dimension of $S$ is $\leq
  i-1$. When $S$ is empty, this is immediate; thus we can assume
  that $S$ is non-empty.

  By definition, $V\inter U=S$, therefore $V\inter\RR^n$ is non-empty;
  note also that $S$ is bounded since we assume that $V\inter \RR^n$
  is bounded. 

  Combining Lemma \ref{lemma:polarnonempty} and ${\sf N}_1$, we
  conclude that $\crit(\pi_i, V_\eps)$ has Krull dimension $i-1$.
  Lemma~\ref{lemma:dim:limites} implies that $\limeps(\crit(\pi_{i},
  V_\eps))$ has Krull dimension $\leq i-1$. We deduce that $\left
    (\limeps(\crit(\pi_i, V_\eps))\right )\inter \RR^n$ has real
  dimension $\leq i-1$. Consequently, $\left (\limeps(\crit(\pi_i,
    V_\eps))\right )\inter U$ has real dimension $\leq i-1$.  By
  assumption, $\left (\limeps(\crit(\pi_i, V_\eps))\right )\inter U =S$;
  we conclude that $S$ has real dimension $\leq i-1$.

  Next, we assume that the real dimension of $S$
  is $\leq i-1$; we prove below that this implies that
  $\left (\limeps(\crit(\pi_i, V_\eps))\right ) \inter U=S$.

  The inclusion $\left (\limeps(\crit(\pi_i, V_\eps))\right ) \inter
  U\subseteq S$ follows from the fact that 
  $\crit(\pi_i, V_\eps)\subseteq \Veps$,
  $\limeps\Veps \subseteq V$,
  and $V\inter U\subset S$.

  It remains to prove the inverse inclusion, that is 
  $S \subseteq \left (\limeps(\crit(\pi_i, V_\eps))\right ) \inter U$.

  If $S$ is empty, then it is immediate that 
  $S \subseteq \left (\limeps(\crit(\pi_i, V_\eps))\right ) \inter U$.
  In the sequel, we assume that $S$ is not empty. 
  
  Let $\x=(\alpha_1,\ldots, \alpha_n)\in S$; since $S=U\inter V$ by
  definition, we have $\x\in U$ and $\x\in V\inter \RR^n$. Thus, we need
  to prove that there exists a point $\x_\eps\in \crit(\pi_i,
  V_\eps)$, such that $\limeps \x_\eps= \x$. This is what we do below.

  \medskip Recall that we have assumed that the real dimension of $S$
  is $\leq i-1$. Since $V \inter U= S$, we deduce that the real
  dimension of $V\inter U$ is $\leq i-1$. It follows that the local real
  dimension of $V\inter U$ at $\x$ is at most $i-1$. We claim that the
  real dimension of $V\inter\RR^n$ at $\x$ is $\leq i-1$. 

  Indeed, let $S_1, \ldots, S_k$ be smooth semi-algebraic sets of
  respective dimension $d_j$ which form a finite partition of
  $V$. The local real dimension of $V\inter \RR^n$ at $\x$
  is the maximum of the $d_j$'s for $j$ such that $\x\in
  \overline{S_j}$. Since $U$ is open, $S_j\inter U$ is either empty or
  smooth of dimension $d_j$. Note also that the semi-algebraic sets
  $S_1\inter U, \ldots, S_k\inter U$ form a partition of $V\inter
  U$. Therefore, the local real dimension of $V\inter \RR^n$ at
  $\x$ is the same as the local real dimension of $V\inter U$ at $\x$
  which is $\leq i-1$ as requested.

  Denoting $\pi_{i-1}(\x)$ by $\x_{i-1}$, we deduce by ${\sf N}_2$
  that $\pi_{i-1}^{-1}(\x_{i-1}) \inter V \inter \RR^n$ is finite.
  Now, let $\tilde f$ be the polynomial obtained after instantiating
  the first $i-1$ variables with the first $i-1$ coordinates of $\x$,
  that is
  \[
  \tilde f = f(X_{i}, \dots, X_n) = f(\alpha_1, \dots, \alpha_{i-i}, X_i, \dots, X_n)
  \enspace .
  \]
  
  We let $\tilde x=(\alpha_i, \ldots, \alpha_n)$ and $\tilde V$
  (resp. $\tilde V_{\eps}\subset \Ce^{n-i+1}$) be the algebraic set
  defined by ${\tilde f}=0$ (resp. $\tilde f=\eps$).
  By assumption, $f$ is non-negative over $\RR^n$.  Therefore, $\tilde
  f$ is non-negative over $\RR^{n-i+1}$.  We also consider the
  canonical projections
  $${\tilde \varphi}_i: (\x_i, \ldots, \x_n)\to \x_i 
  \quad \text{ and }\quad \varphi_i: (\x_1, \ldots, \x_n)\to \x_i.$$
  By applying Lemma~\ref{lemma:locallyfinite} to $\tilde V$ and
  $\tilde f$, there exists $ \tilde \x_\eps\in \crit(\tilde\varphi_i,
  \tilde{V}_\eps)$,
  such that $\limeps\tilde\x_\eps=(\alpha_i, \ldots, \alpha_n)$.  Now
  define $\x_\eps=(\alpha_1, \ldots, \alpha_{i-1}, \tilde\x_\eps)$ and
  $V'_{\eps}=V_\eps\inter \pi_{i-1}^{-1}(\alpha_1, \ldots,
  \alpha_{i-1})$. 


  Since $ \tilde \x_\eps\in \crit(\tilde\varphi_i,\tilde{V}_\eps)$,
  it is immediate that $\x_\eps\in \crit(\varphi_i, V'_{\eps})$,
  and using Lemma~\ref{lemma:pointscritiquesetfibres}
  $\x_\eps \in \crit(\pi_i, V_\eps)$. To summarize,
  we have established $\x_\eps \in \crit(\pi_i, V_\eps)$ and
  $\limeps\x_\eps=\x$ as requested.  This finishes the proof.
\end{proof}

 
The following corollary is a direct consequence of
Proposition~\ref{prop:criterion:dimension} and it is crucial for the
proof of correctness of our algorithm.  Recall that, by convention,
$\crit(\pi_0, V_\eps)$ is the empty set.
\begin{corollary}
  \label{cor:Wi-diff}
  Assume that $f$ satisfies property $\func{N}$ and that $V$ is
  bounded.  Let $d$ be the real dimension of $S$.  Then, $d\geq 0$ if
  and only if there exist an integer $i$ in $\{1, \ldots, n\}$ such
  that $\left (\limeps \crit(\pi_{i}, V_\eps)\right ) \inter U \neq
  \left (\limeps \crit(\pi_{i-1}, V_\eps)\right ) \inter U$ and $d+1$
  is the largest of these integers.
\end{corollary}
\begin{proof}
  Note that the assumptions of
  Proposition~\ref{prop:criterion:dimension} are satisfied. This
  implies that, for $0\leq i \leq n$, $\left (\limeps\crit(\pi_i,
    V_\eps)\right )\inter U=S$ if and only if $d\leq i-1$.  As a
  consequence, we deduce that
$$
S=\left (\limeps\crit(\pi_{d+1}, V_\eps)\right )\inter U=\cdots=\left (\limeps\crit(\pi_{n}, V_\eps)\right )\inter U
$$
{ and } $S\neq \left (\limeps\crit(\pi_{i}, V_\eps)\right )\inter U${ for } $1\leq i \leq d$.
This concludes the proof. 
\end{proof}

\section{Algorithm and proof of Theorem \ref{thm:main}}
\label{sec:Algo}

{\bf There is an error in the main theorem and the stated bound does not hold.}

\subsection{Main algorithm}

The input of our algorithm is two sets of polynomials $F = (f_1,
\ldots, f_{p})$ and $G=(g_1, \ldots, g_s)$ in $ \ZZ[X_1, \ldots,
X_n]$.  The output is the real dimension of the semi-algebraic set
$\mathscr{S}\subset \RR^n$ defined by 
\[
f_1=\cdots=f_p=0, \quad g_1>0, \ldots,g_s>0  \enspace.
\]

The main idea of the algorithm is to perform symbolic manipulations to
define a polynomial $f$ with coefficients in $\RR[\zeta]$ such that
\begin{itemize}
\item $f$ and the semi-algebraic set $S\subset \RRz^{n+1}$ defined
  by $f=0$ and $g_1>0, \ldots, g_s>0$ satisfy the assumptions of
  Corollary~\ref{cor:Wi-diff}, and
\item the real dimension of $S$ is the same as the real
  dimension of $\mathscr{S}$.
\end{itemize}
Then, we apply the results of the previous section
(Section~\ref{sec:geom-stmt}) using as ground field (that
is the field where the coefficients of the input polynomials belong
to) $\RRz$.

Let us recall the notations that we use: $S\subset \RRz^{n+1}$
is the semi-algebraic set defined by $f=0, g_1>0, \ldots, g_s>0$,
$U\subset \RRz^{n+1}$ is the open semi-algebraic set defined by
$g_1>0, \ldots, g_s>0$, and $V\subset \CCz^{n+1}$,
resp. $V_\eps\subset \CCze^{n+1}$, is the algebraic set defined by
$f=0$, resp. $f=\eps$.

From Corollary~\ref{cor:Wi-diff}, $d\geq 0$ if and only if there
exists an integer $i$ in $\{1, \ldots, n\}$ such that $$\left (\limeps
  \crit(\pi_{i}, V_\eps)\right ) \inter U \neq\left ( \limeps
  \crit(\pi_{i-1}, V_\eps)\right ) \inter U \,,$$ and $d+1$ is the
largest of these integers (notice that by convention $\crit(\pi_0,
V_\eps)=\emptyset$; see Section \ref{sec:prelim}).

\paragraph*{Subroutines.}
We need three subroutines.  The first one is called
\func{Random}; it takes as input an integer $n$ and returns a randomly
chosen $n \times n$ matrix in $\GL_n(\ZZ)$.


The second routine \func{IsEmpty} takes as input a polynomial $H$ and
a set of polynomials ${\cal G}=(G_1, \ldots, G_s)$ in $\ZZ[\zeta,\eta,
\eps][X_1, \ldots, X_n]$,
where $\zeta$, $\eta$, and $\eps$ are infinitesimals. 
Let $V \subset \CCzhe^n$ be the algebraic set defined by $H=0$ 
and $U \in \RRzhe^n$ the open semi-algebraic set defined 
by $G_1>0, \ldots, G_s>0$.
The subroutine decides if the semi-algebraic set 
$V \inter U$ is empty or not.

Let $\delta$ be the maximum of the degrees of the monomials in $\zeta,
\eta, \eps, X_1, \ldots, X_n$ appearing in $H$ and ${\cal G}$. If
$\ZZ=\mathbb{Z}$, then we denote by $\tau$ the maximum bit size of the
integers appearing in $H$ and ${\cal G}$.  The algorithm returns
\func{True} if the semi-algebraic set $V \inter U$ is empty, and
\func{False} otherwise. It is based on \cite[Algorithm~13.1]{BaPoRo06}.

\begin{lemma}
  \label{lem:subroutine-IsEmpty}
  Using the above notations, algorithm \func{IsEmpty}$(H, {\cal G})$ decides
  if the semi-algebraic set defined by $H=0$ and $G_1>0, \ldots,
  G_s>0$ is empty or not within $(s \,\delta)^{O(n)}$ arithmetic
  operations in $\ZZ$.
  When $\ZZ=\mathbb{Z}$, then the Boolean complexity of  
  \func{IsEmpty} is $\tau \, (s \, \delta)^{O(n)}$.
\end{lemma}
\proofatend
  We decide the emptiness of $V \inter U$ using the algorithm from \cite[Proposition~2.2]{v-ld-jsc-1999}.
  The cost is $(s\, \delta)^{O(n)}$ operations in $\ZZ$.
  The Boolean complexity bound follows by combining \cite[Proposition~2.2]{v-ld-jsc-1999}
  with \cite[Algorithm~13.1]{BaPoRo06}.
\endproofatend
When \func{IsEmpty} is called by the main algorithm,
the input polynomials have coefficients in $\ZZ[\zeta]$.


The third subroutine, \func{DisjointPolar}, takes as input a
polynomial $f\in \ZZ[\zeta][X_1, \ldots, X_{n+1}]$, 
$G=(g_1, \ldots, g_s)\subset \ZZ[X_1, \ldots, X_n]$,
$Q \in \ZZ[\zeta][X_1, \ldots, X_{n+1}]$,
 and an integer $i\in \{1, \ldots,
n\}$ (where $\zeta$ is an infinitesimal that we manipulate as a
variable). 
The polynomial $Q$ defines a sphere in $\RRz^{n+1}$ that strictly contains the 
semi-algebraic set $S$. For example 
$Q = (\zeta (X_1^2 + \dots + X_n^2 +X_{n+1}^2)-2)$.

The routine $\func{DisjointPolar}(f, G, Q,i)$ returns \func{True} if
$$\left (\limeps\crit(\pi_i, V_\eps)\cap U\right )\neq \left (\limeps\crit(\pi_{i-1}, V_\eps)\right )\cap U,$$
and \func{False} otherwise.  Below, we denote by $\delta$ the maximal
degree of the monomials in $\zeta, X_1, \ldots, X_n$ appearing in $f$,
$G$ and $Q$.

The routine \func{DisjointPolar} is described 
Section~\ref{sec:disjoint-pair}, where we also prove the following
Lemma.

\begin{lemma}\label{lem:disjoint-polar}
  Let $f\in \ZZ[\zeta][X_1, \ldots, X_{n+1}]$, $G=(g_1, \ldots,
  g_s)\subset \ZZ[X_1, \ldots, X_n]$, $B\in \ZZ[\zeta][X_1, \ldots,
  X_{n+1}]$, and $i\geq 0$.  Let $S$ be the semi-algebraic set defined
  by $\{f=0 \wedge G > 0\}$, and let $Q \leq 0$ define a ball that
  strictly contains $S$.  Assume that $f$ satisfies $\sf N$ and is
  non-negative over $\RRz^n$ and that $\left (\limeps\crit(\pi_i,
    V_\eps)\right )\inter U$ is not empty.
There exists an algorithm \func{DisjointPolar} which decides if
  $\left (\limeps\crit(\pi_i, V_\eps)\right )\cap U\neq \left
    (\limeps\crit(\pi_{i-1}, V_\eps)\right )\cap U$ within $(s \,
  \delta)^{O(n)}$ operations in $\ZZ$.

  If $\ZZ=\mathbb{Z}$ and the maximum bit size of the integers in the coefficients of
  the input polynomials is $\tau$, then the Boolean complexity of the
  algorithm is $\tau\, (s \,\delta)^{O(n)}$.
\end{lemma}

We can now describe our main algorithm \func{ComputeRealDimension}. 

{\small
\begin{algorithm2e}[H]
 \caption{\FuncSty{{\sf ComputeRealDimension}}($F$, $G$)}
  \label{alg:RealDim}
  \SetKw{RET}{{\sc return}}
  \KwIn{$F =(f_1, \dots, f_p) $ and $G=(g_1, \ldots, g_s)$ in $\ZZ[X_1, \ldots, X_n]$}
  \KwOut{The real dimension of the semi-algebraic set defined by $f_1=\cdots=f_p=0, g_1>0, \ldots, g_s>0$}
  
  \BlankLine
  $f_0  \leftarrow \sum_i f_{i}^2$ \; \nllabel{alg:sos}
  \BlankLine

  \lIf{ {\sf IsEmpty($f_0$, $G$)} }{ \nllabel{alg:empty}
    \RET $-1$ \;
  }
  
  \BlankLine
  $f  \leftarrow f_0 + (\zeta (X_1^2 + \dots + X_n^2 +X_{n+1}^2)-1)^2$ \; \nllabel{alg:unbd}
  \BlankLine
  $ Q \leftarrow (\zeta (X_1^2 + \dots + X_n^2 +X_{n+1}^2)-2)$ \;
  \BlankLine

  $\mA \leftarrow {\sf Random}(n)$ \; \nllabel{alg:A-random}

  \BlankLine
  $f^{\mA}(\X) \leftarrow f( \mA\,\X)$ and   $G^{\mA}(\X) \leftarrow G( \mA\,\X)$ and $Q^{\mA}(\X) \leftarrow Q( \mA\,\X)$\; \nllabel{alg:change-var}



  \BlankLine

  \For{$n \geq i \geq 1$ }{
    
    \BlankLine
    
  

    \lIf{ {\sf DisjointPolar($f^\mA$, $G^\mA$, $Q^\mA$, $i$)} }{ \nllabel{alg:check-diff}
      \RET $i-1$ \;
    }
  }
\end{algorithm2e}
}

The proof of Theorem \ref{thm:main} consists of proving the
correctness and complexity estimate of \func{ComputeRealDimension}.

\paragraph*{Correctness.}

Let $d$ be the real dimension of the semi-algebraic set $\mathscr{S}\subset
\RR^n$ defined by $f_1=\cdots=f_p=0, g_1>0, \ldots, g_s>0$.

At Step~\ref{alg:empty} we test whether the semi-algebraic set is
empty or not.  In the sequel we assume that the semi-algebraic
set is not empty, and so $d\geq 0$.

We also denote by $V\subset \CCz^{n+1}$ the algebraic set defined by $f=0$
where $f$ is defined at Step \ref{alg:unbd}.  Assume for the moment
that the semi-algebraic set $S$ has real dimension $d$ and that $V\cap
\RRz^n$ is bounded.  At Step \ref{alg:change-var}, the real dimension
of the semi-algebraic set $S^\mA$, defined by $f^{\mA}=0$ and
$g_1^{\mA}>0, \ldots, g_s^{\mA}>0$, is also $d$.  Since $\mA$ is randomly
chosen we can assume that it lies in the non-empty Zariski open set
$\mathscr{O}$ defined in Proposition~\ref{prop:normalisgeneric}.
Therefore, $f^\mA$ satisfies property $\sf N$ (see
Definition~\ref{def:generic-position}) and $V^\mA\cap\RRz^{n+1}$ (and
consequently $S^\mA$) is bounded.  In other words, all the assumptions
of Corollary \ref{cor:Wi-diff} are satisfied.

In the for-loop, starting with $i=n$, the algorithm looks for the largest integer $i$ such that 
$$
\left (\limeps\crit(\pi_i, V^\mA_\eps)\right )\cap U^\mA 
\neq 
\left (\limeps\crit(\pi_{i-1}, V^\mA_\eps)\right )\cap U^\mA
\enspace.
$$

Each time we enter in the loop $\left (\limeps\crit(\pi_i,
  V^\mA_\eps)\right )\cap U^\mA$ is not empty (and actually equals
$S^\mA$).  The other assumptions of Lemma \ref{lem:disjoint-polar} are
obviously satisfied.  Also, by Corollary \ref{cor:Wi-diff} and Lemma
\ref{lem:disjoint-polar}, we deduce that when $d\geq 0$, at Step
\ref{alg:check-diff} the algorithm will return $d$.

\medskip To finish the proof it remains to establish that $S$ is
bounded and its real dimension is the same as the real dimension of
$\mathscr{S}$.

We start with the boundedness statement. The polynomial $f$ is the sum
of the squares of $f_1, \ldots, f_p$ and the square of the polynomial
$\zeta (X_1^2 + \dots + X_n^2 +X_{n+1}^2)-1$. The set of roots in
$\RRz^{n+1}$ of this latter polynomial is the $n$-dimensional sphere,
$\mathbb{S}_{n+1}$, with center at the origin and radius
$1/\sqrt{\zeta}$. It is straightforward that $V\cap\RRz^{n+1}$ is
bounded and since $S=V\cap U$, we deduce that $S$ is bounded.

Now, we prove that $S$ has the same real dimension, $d$,
as the semi-algebraic set $\mathscr{S}$. 
The proof consists of two steps. 
We prove consecutively that 
\begin{enumerate}
\item[{\em (i)}] $\ext(\mathscr{S}, \RRz)\subset \RRz^{n}$ has
real dimension $d$, and 
\item[{\em (ii)}] $S$ has the same real dimension
  as $\ext(\mathscr{S}, \RRz)\subset \RRz^{n}$.
\end{enumerate}

Statement {\em (i)} is a straightforward consequence of the Lemma below. 
\begin{lemma}
  \label{lem:S-eq-Sz}
  Let ${\cal S}$ be a semi-algebraic subset of $\RR^n$ and consider its
  extension $\ext({\cal S}, \RRz)$.  Then ${\cal S}$ and $ \ext({\cal S},
  \RRz^n)$ have the same real dimension.
\end{lemma}
\begin{proof}
  Let $d$ be the real dimension of $ {\cal S}$ and $d'$ be the real
  dimension of $\ext({\cal S}, \RRz)$. Then there is an injective map
  $\varphi: {\cal S} \rightarrow [0,1]^d$.  We consider the extension
  of $\varphi$, that is $\ext(\varphi,\RRz)$. It is also injective
  \cite[Exercise~2.17]{BaPoRo06} and by the definition of the real
  dimension (see Definition~\ref{def:dim}) we deduce that $d'\leq d$.
  Now remark that $\lim_{\zeta\rightarrow 0} \ext({\cal S},
  \RRz)={\cal S}$. By Lemma \ref{lemma:dim:limites} we deduce that
  $d\leq d'$ which finishes the proof.
\end{proof}

To prove {\em (ii)}, we consider the restriction of
$\pi_n$ to $S$. It is a semi-algebraic and injective function, since for all $(\x,
x_{n+1})\in \RRz^n\times \RRz$ with $\x\in \ext(\mathscr{S}, \RRz)$
and $(\x, x_{n+1})\in \mathbb{S}_{n+1}$,
$x_{n+1}=\sqrt{\frac{1}{\zeta}-||\x||^2}$. 
Moreover, 
$\pi_n(S)=\ext(\mathscr{S}, \RRz)$.

By \cite[Proposition~5.29]{BaPoRo06} we deduce that the real dimension
of $S$ is equal to  the real dimension of $\ext(\mathscr{S}, \RRz)$;
which is $d$ by {\em (i)}.  \hfill\qed

\paragraph*{Complexity Analysis.}
The complexity estimate is straightforward by
Lemmata~\ref{lem:subroutine-IsEmpty} and \ref{lem:disjoint-polar}.
The first call to \func{IsEmpty} costs at most $(s \, D)^{O(n)}$
arithmetic operations.  Next, the algorithm calls $n$ times the sub-routine
{\sf DisjointPolar} with input polynomials of total degree
$O(D)$. Each call costs $(s\,D)^{O(n)}$ operations in $\ZZ$.  Since
all functions in the complexity class $n(s\,D)^{O(n)}$ lie in the
complexity class $(s\,D)^{O(n)}$, we deduce that
\func{ComputeRealDimension} runs within $(s\,D)^{O(n)}$ arithmetic
operations in $\ZZ$.

This is a probabilistic bound because of the random change of
coordinates that we apply to $f$ and $G$ at Step~\ref{alg:change-var}
to ensure property \func{N} (Definition~\ref{def:generic-position}).

\subsection{The subroutine {\sf DisjointPolar} and Proof of Lemma \ref{lem:disjoint-polar}}
\label{sec:subroutines}

In order to describe \func{DisjointPolar}, we need to recall some
fundamental algorithmic specifications and complexity results in
computational real algebraic geometry. Sometimes we need to prove some
statements which are quite folklore in the area; proofs of these facts
are given in the Appendix.

\paragraph*{Quantifier Elimination over the reals.}
\label{sec:qe}
Let $\X=(X_1, \ldots, X_n)$ and $\Y = (Y_1, \dots, Y_t)$ be finite
sequences of variables, $\Omega$ be the existential quantifier $\exists$ or
the universal quantifier $\forall$ and $\mathcal{P}(\X, \Y)$ be a
Boolean function of $s$ atomic predicates $H_i(\X, \Y) \rhd_i 0$,
where $\rhd_i~\in~\{>~, <~, =~\}$ and $H_i\in \ZZ[\X, \Y]$, for $1
\leq i \leq s$.  We let $\delta=\max(\deg(H_i), 1\leq i \leq s)$ and
when $\ZZ=\mathbb{Z}$, $\tau$ is the maximum of the bit size of the
coefficients in the $H_i$'s

One-block quantifier elimination over the reals consists in computing
a quantifier-free formula which is equivalent to the first order
quantified formula $\Phi:  (\Omega \, \X \in \RR^{n}) \, \mathcal{P}(\X, \Y)$.

The following theorem is a simplification of the general purpose
quantifier elimination algorithm in \cite[Algorithm~14.5 and
Theorem~14.16]{BaPoRo06}.

\begin{theorem}[One Block Quantifier Elimination over the reals]
  \label{th:qe}
  There exists an algorithm \\$\func{OneBlockQuantifierElimination}$
  which takes as input $\Omega, {\cal P}, \X, \Y$ and returns a quantifier
  free formula $\Psi$ of the form 
  $ \bigvee_{i=1}^{I} \bigwedge_{j=1}^{J_i} h_{i,j} \rhd_{i,j} 0 $,
  where $h_{i,j} \in \ZZ[Y_1, \dots, Y_t]$ and 
  $\rhd_{i,j} \in \{ >,= \}$, that is equivalent to $\Phi$, such that
  $$
  I \leq s^{(t+1)} \,n\, \delta^{(t+1) O(n)}, \quad J_i \leq s^{(n+1)} \, \delta^{O(n)}, 
  \quad \deg(h_{i,j}) \leq \delta^{O(n)}
   \enspace .
  $$
  If $\ZZ=\mathbb{Z}$, then the maximum bit size of the coefficients of
  $h_{i,j}$ is bounded by $\tau \delta^{(t+1)O(n)}$.
  The above transformation requires $s^{(t+1)(n+1)} \delta^{(t+1)O(n)}$
  arithmetic operations in $\ZZ$.
  The Boolean complexity of the algorithm is $\tau s^{(t+1)(n+1)}
  \delta^{(t+1)O(n)}$.
\end{theorem}

\paragraph*{Limits of semi-algebraic sets in $\RRzhe$ when $\eps\to 0$.}
\label{sec:limit}
The routine \func{UnivariateLimit} takes as input a quantifier-free
formula $\Psi$, which is a disjunction of conjunctions of polynomial
equations/inequalities in $\ZZ[\zeta,\eta,\eps][Z]$, and the
infinitesimal $\eps$. This formula defines a semi-algebraic set ${\cal
  S}_{\eps}$ in $\RRzhe$.  It outputs a semi-algebraic description of
$\limeps{\cal S}_{\eps}$. This is based on quantifier elimination over
the reals \cite[Theorem~14.16]{BaPoRo06}.

\begin{lemma}
  \label{lem:subroutinelimit}
  Let $\delta$ be the maximum of the degrees of the momonials in
  $ \zeta, \eta, \eps$ of the polynomials in $\Psi$ and $\ell$ be the
  number of polynomials in $\Psi$.  One can compute
  $\func{UnivariateLimit}(\Psi,\eps)$ within $(\ell \,\delta)^{O(1)}$
  arithmetic operations in $\ZZ$.

  When $\ZZ=\mathbb{Z}$ and $\tau$ is a bound on the bit size of the
  integers of the coefficients in $\Psi$, then ${\sf
    UnivariateLimit}(\Psi, \eps)$ runs within $\tau (\ell
  \,\delta)^{O(1)}$ bit operations.
\end{lemma}
\proofatend
   The quantifier free formula of the input represents the following semi-algebraic set
  \begin{displaymath}
    {\cal S}_{\eps} = \{ z \in \RRzhe \,|\, \Psi \} \enspace ,
  \end{displaymath}
  where $\Psi$ is a Boolean formula whose atoms are polynomials in $Z$
  with coefficients in $\RRzhe$.
  We can express the limit of ${\cal S}_{\eps}$ as $\eps \rightarrow 0$, 
  ${\cal S}$, using the language of the first order theory over
  the reals \cite[Sec.~3.1]{BaPoRo06}, that is
  \begin{displaymath}
    {\cal S} = \limeps {\cal S}_{\eps} = 
    \{
    z' \in \RRz\langle\eta\rangle \,|\, 
    (\forall r > 0) (\exists \eps_0 > 0)  (\forall \eps) \,
    (\forall z) \, z\in {\cal S}\wedge 
    (0 < \eps < \eps_0 \Rightarrow \| z - z' \|^2 < r^2)
    \}
    \enspace .
  \end{displaymath}
  In this way, to compute the limit it suffices to eliminate the
  quantifiers from the previous formula. 

  For the elimination process, we treat $\zeta$ as a variable. To see
  that this is valid let $\Phi$ be the Boolean formula the describes $
  z\in {\cal S}\wedge (0 < \eps < \eps_0 \Rightarrow \| z - z' \|^2 <
  r^2)$.  Then
  \begin{displaymath}
    A_{\zeta,\eta} = \{ (z, z', r, \eps_0, \eps) \in \RRz\langle\eta\rangle^5 \,|\, \Phi \} \enspace.
  \end{displaymath}
  If we let $\zeta = t$ and $\eta=h$ in all the polynomials in $\Phi$, then 
  we get the set 
  \begin{displaymath}
    A = \{  \{ (z, z', r, \eps_0, \eps, t, h) \in \RR^6 \,|\, \Phi \}\} 
    \enspace.
  \end{displaymath}
  In this way 
  \begin{displaymath}
    \pi_{-t}( \ext(A, \RRz\langle\eta\rangle) \inter \{t=\zeta, h=\eta\}) = A_{\zeta,\eta} \subset \RRz\langle\eta\rangle^5
    \enspace, 
  \end{displaymath}
  where $\pi_{-th}$ is the projection 
  $(z, z', r, \eps_0, \eps, t,h)  \mapsto (z, z', r, \eps_0, \eps)$.

  We notice that ${\cal S}$ is defined using a constant number of free
  variables, ($\zeta$, $\eta$, and $Z'$), a constant number of quantified
  variables, ($z$, $r$, $\eps_0$, $\eps$), and a constant number of
  quantifier alternations.  Following \cite[Theorem~14.16]{BaPoRo06},
  after we perform quantifier elimination, the output is a
  quantifier-free formula $\bigvee_{i=1}^{I}\bigwedge_{j=1}^{J_i}
  h_{i,j} \rhd_{i,j} 0$, where $h_{i,j} \in \ZZ[\zeta,\eta][Z']$ have
  degree at most $\delta^{O(1)}$, $I \leq (\ell \, \delta)^{O(1)}$,
  $J_i \leq (\ell \, \delta)^{O(1)}$, and the complexity of the
  quantifier elimination procedure is $(\ell \, \delta)^{O(1)}$
  arithmetic operations in $\ZZ$.

  If $\ZZ = \mathbb{Z}$ and the input polynomials have coefficients of
  maximum bit size $\tau$, then the Boolean complexity is 
  $\tau \,(\ell\, \delta)^{O(1)}$.
\endproofatend

\paragraph*{Computing limits of critical points.}
The routine \func{LimitsOfCriticalPoints} takes as input a polynomial
$H$ and a set of polynomials ${\cal G}=(G_1, \ldots, G_s)$ in
$\ZZ[\zeta][X_1, \ldots, X_n]$. We denote by $V_\eps\subset \CCze^n$
the algebraic set defined by $H-\eps=0$, by $U\subset \RRz^n$ the open
semi-algebraic set defined by $G_1>0, \ldots, G_s>0$ and by $\delta$
the maximum of the degrees of the monomials in $\zeta, X_1, \ldots,
X_n$ appearing in $H$ and ${\cal G}$. If $\ZZ=\mathbb{Z}$, then we
denote by $\tau$ the maximum bit size of the integers appearing in $H$
and ${\cal G}$. It is based on \cite[Algorithm~13.1]{BaPoRo06} and
\cite[Algorithm~11.20]{BaPoRo06}.
\begin{lemma}\label{lemma:limitsoffiniteset}
  Assume that $H$ is non-negative over $\RRz^n$ and that $\crit(\pi_1,
  V_\eps)$ is finite.  There exists an algorithm
  \func{LimitsOfOfCriticalPoints} which takes as input $H$ and ${\cal
    G}$ as above, and returns \func{True} if $\left (\limeps\crit(\pi_1,
  V_\eps)\right )\inter U$ is non-empty else it returns \func{False}, within $(s
  \,\delta)^{O(n)}$ arithmetic operations in $\ZZ$.

  When $\ZZ=\mathbb{Z}$ and $\tau$ is the maximum bit size of the coefficients,
  then the Boolean complexity of  
  \func{LimitsOfCriticalPoints} is 
  $\tau \, (s \, \delta)^{O(n)}$.
\end{lemma}
\proofatend
  By assumption, $\crit(\pi_1, V_\eps)$ is finite and $H$ is
  non-negative over $\RRz^n$. Thus, to compute $\left (\limeps \crit(\pi_1,
  V_\eps)\right )\inter U$ it is sufficient to
  \begin{enumerate}
  \item[(1)] compute sample points in each connected component of
    $\crit(\pi_1, V_\eps)\cap\RRze^n$; they will be encoded with a
    rational parametrization with coefficients in $\ZZ[\zeta,\eps]$
$$q(T)=0, X_1=q_1(T)/q_0(T), \ldots, X_n=q_n(T)/q_0(T)
;$$
  \item[(2)] use this parametrization to compute their limits and make
    the intersection with $U$.
  \end{enumerate}


  The above parametrization is obtained using
  \cite[Algorithm~13.1]{BaPoRo06} with input $H-\eps=\frac{\partial
    H}{\partial X_2}=\cdots=\frac{\partial H}{\partial
    X_n}=0$. Following {\em mutatis mutandis} the same reasoning as in
  \cite[Proposition~2.2]{v-ld-jsc-1999}, we deduce that the arithmetic
  cost in $\ZZ$ of this step is $(s\, \delta)^{O(n)}$.

  Step (2) is performed using with \cite[Algorithm~11.20]{BaPoRo06}
  (Removal of an infinitesimal) for compute the limit of the sample
  points as $\eps \rightarrow 0$ and sign determination algorithms for
  univariate polynomials to obtain the intersection with $U$ (see
  \cite[Algorithm 10.13]{BaPoRo06}).  The overall complexity is $(s
  \,\delta)^{O(n)}$ operations in $\ZZ$ \cite[Chapters 10 and
  11]{BaPoRo06}.
  
  The Boolean complexity bound is straightforward from the above
  reasoning.
\endproofatend

\paragraph*{The routine \func{IsRealizable}.}
\label{sec:realizable}

A typical call of this routine is $\func{IsRealizable}( \Phi)$, where
$\Phi$ is a union of disjunctions of univariate polynomials in $Z$
with coefficients in $\ZZ[\zeta,\eta]$.

The subroutine calls Algorithm~10.13 from \cite{BaPoRo06} to compute
all realizable sign condition of the polynomials and checks if there
is at least one sign condition that is compatible with $\Phi$.  In
this case it returns \func{True}, otherwise it returns \func{False}.

\begin{lemma}\label{lemma:IsRealizable}
  Let $\Phi$ be a union of disjunctions of univariate polynomials in
  $Z$ with coefficients in $\ZZ[\zeta,\eta]$ such that all monomials in
  $Z,\zeta,\eta$ appearing in $\Phi$ have degree at most $\delta$.
  
  The complexity of $\func{IsRealizable}(\Phi)$ is $\ell \,\delta^{O(1)}$
  operations in $\ZZ$.
  If $\ZZ = \mathbb{Z}$ and the maximum bit size of the coefficients of the
  polynomials of the input is $\tau$, then the Boolean complexity is
  $\tau\, \ell \, \delta^{O(1)}$.
\end{lemma}
\proofatend
  The subroutine is based on Algorithm~10.13 from \cite{BaPoRo06}.  If
  the degree of the polynomials is at most $\delta$ and if there are
  at most $\ell$ polynomials, then the total cost is $\ell \,\delta^{O(1)}$
  arithmetic operations in $\ZZ[\zeta,\eta]$ \cite{BaPoRo06}. 

  The operations that we need are computations of subresultant sequences for univariate
  polynomials in $Z$ with coefficients in $\ZZ[\zeta,\eta]$.  
  The coefficients of the polynomials in the sequence have degree at most $\delta^{O(1)}$ in $\zeta$ and $\eta$
  \cite[Proposition~8.49]{BaPoRo06}.
  Therefore
  the complexity of the algorithm is also $\ell \,\delta^{O(1)}$,
  when we count operations in $\ZZ$.
  
  The Boolean complexity is $\tau\, \ell \, \delta^{O(1)}$ and it is
  due to the bit size of the integers that appear in the subresultant
  sequence.
\endproofatend
\paragraph*{The routine \func{DisjointPolar}.}
\label{sec:disjoint-pair}
Let $f\in \ZZ[\zeta][X_1, \ldots, X_n]$, a set of polynomials $G
=(g_1,\ldots, g_s)\subset \ZZ[X_1, \ldots, X_n]$, a polynomial $Q\in
\ZZ[\zeta][X_1, \ldots, X_n]$ and an integer $i$, such that $1 \leq i
\leq n$. 

The open semi-algebraic set of $\RRz^n$ defined by $g_1 >0, \dots, g_s
> 0$ is denoted by $U $.  We introduce another infinitesimal
$\eta$, with $\zeta >\eta >0$, and we consider the semi-algebraic set
$U_\eta\subset \RRz\langle \eta\rangle^n$ defined by $g_1\geq \eta,
\ldots, g_s\geq \eta$. We also denote the latter inequalities by $G
\geq \eta$. Note that $U_\eta$ is closed.

In the sequel, we denote $\RRz\langle \eta\rangle$ by $\R$ and
$\CCz\langle \eta\rangle$ by $\C$. 

Finally, we introduce one more infinitesimal $\eps$ with
$\zeta > \eta > \eps$.
Let $V_\eps\subset \Ce^n$ be the algebraic variety 
defined by $f-\eps=0$. To make the notation simpler, $\ext(U_\eta, \Re)$
will be denoted by $U_{\eta, \eps}$ and $\ext(B, \Re)$ will be denoted
by $B_\eps$.

We assume that
\begin{itemize}
\item $f$ satisfies property \func{N}
(Definition~\ref{def:generic-position}) and is non-negative over
$\RRz^n$;
\item that $Q\leq 0$ defines a ball $B$ that strictly contains the
  semi-algebraic set defined by $f=0$ and $g_1>0, \ldots, g_s>0$.
\item and the solution set of $g_1>0, \ldots, g_s>0$ has a non-empty
  intersection with $\crit(\pi_i, V_\eps)$;
\end{itemize}

\noindent 
Then, $\func{DisjointPolar}(f, G, Q, i)$ returns {\sf True} if $\left
  (\limeps \crit( \pi_i, V_{\eps})\right ) \inter U \not=\left (
  \limeps \crit( \pi_{i-1}, V_{\eps})\right ) \inter U$, and {\sf
  False} otherwise. Before giving a detailed description, we briefly
give a geometric view of the operations performed by
\func{DisjointPolar}.

The algorithm works as follows. We consider separately the case $i=1$ (Step~\ref{alg:ifieq1}).  By
convention $\crit(\pi_0, V_{\eps})$ is the empty set.  Therefore,
$\left (\limeps \crit( \pi_{0}, V_{\eps})\right ) \inter U =
\emptyset$ and if $i=1$, then it suffices to check whether $\left
  (\limeps \crit(\pi_1, V_\eps)\right )\inter U$ is empty or not.

When $i\geq 2$, the algorithm starts by testing if $\crit(\pi_{i-1},
V_\eps)\inter U_{\eta, \eps}$ is empty. We will see that when this is
the case, it implies that $(\limeps\crit(\pi_{i-1}, V_\eps))\inter U$
is empty. Since by assumption $(\limeps\crit(\pi_{i-1}, V_\eps))\inter
U$ is not empty, the routine returns {\sf True}.

When $\crit(\pi_{i-1}, V_\eps)\inter U_{\eta, \eps}$ is empty, the
routine constructs a formula that defines the following semi-algebraic
set
\begin{displaymath}
  \begin{aligned}
    A_{\eps} = \{ (\x, \y, z) \in \Re^{2n+1} &\quad |&& 
    \x\in \left (\crit(\pi_{i}, V_\eps) - \crit(\pi_{i-1}, V_\eps) \right )\inter U_{\eta,\eps}\inter B_\eps,  \\
    &&&\y\in \crit(\pi_{i-1}, V_\eps)\inter U_{\eta,\eps}\inter B_\eps,\, ||x-y||=z > 0\}
     \enspace .
  \end{aligned}
\end{displaymath}

Next, we use quantifier elimination to compute a semi-algebraic
description of the semi-algebraic set $\limeps\pi_Z(A_\eps)\subset \R$ (where $\pi_Z$ is the projection
on the $Z$-coordinate). We will prove that
$$\left (\limeps \crit( \pi_i, V_{\eps}) \right )\inter U \neq
\left (\limeps \crit( \pi_{i-1}, V_{\eps})\right ) \inter U
\Longleftrightarrow \limeps\pi_Z(A_\eps) \inter \{z>0\}\neq
\emptyset\enspace .$$

{\small
\begin{algorithm2e}[h]
  \caption{\FuncSty{\sf DisjointPolar}($f, G, Q, i$)}
  \label{alg:DisjointPolar}
  \SetKw{RET}{{\sc return}}
  \SetKwInput{KwData}{Assumptions}
  \KwIn{
    $f \in \ZZ[\zeta][X_1, \ldots, X_n]$, 
    $G=(g_1, \ldots, g_s)\subset \ZZ[X_1, \ldots, X_n]$, $Q\in \ZZ[\zeta][X_1, \ldots, X_n]$,
    and 
    $i$  }
  \KwOut{{\sf True} if $\left (\limeps \crit( \pi_i, V_{\eps}) \right )\inter U \not= \left (\limeps \crit( \pi_{i-1}, V_{\eps}) \right ) \inter U$.
    {\sf False} otherwise.}

  \KwData{$f$ satisfies $\sf N$ and  is non-negative over $\RRz^n$,
    $Q \leq 0$ defines a ball strictly containing $S$,
    $\left (\limeps\crit(\pi_i, V_\eps)\right )\inter U \not= \emptyset$,
    and $1 \leq i \leq n$}
  
  \BlankLine
  \lIf{$i=1$} {\nllabel{alg:ifieq1} 
     \RET $\func{LimitsOfCriticalPoints}(f, G)$
  }
    
  \BlankLine
  $P_1 \leftarrow \{ \left (f(\X)-\eps\right )^2+\left ( \frac{\partial f(\X)}{\partial X_{i+1}}\right )^2+\cdots+ \left ( \frac{\partial f(\X)}{\partial X_n}\right )^2=0 \wedge G \geq \eta \wedge Q\leq 0\}$ \; \nllabel{alg:F-def}
  
  \BlankLine 
  $P_2 \leftarrow \func{Subs}(\X=\Y,\{\left (f(\X)-\eps\right )^2+\left ( \frac{\partial f(\X)}{\partial X_{i}}\right )^2+\cdots+ \left ( \frac{\partial f(\X)}{\partial X_n}\right )^2=0  \wedge G \geq \eta \wedge Q\leq 0\})$ \; \nllabel{alg:G-def}

    \lIf{ $\func{IsEmpty}(P_2)$ }{ 
      \RET {\sf False} \; \nllabel{alg:emptytest}
    }

  \BlankLine  
  $\Phi \leftarrow 
  \, 
  [\, P_1(\X)  \,\wedge\,  P_2(\Y)  \,\wedge\, 
  \frac{\partial f (\X)}{\partial X_{i}} \not= 0
  \,\wedge\,  \| \X- \Y\|^2 > Z^2  \,\wedge\,  Z > 0 \,]$ \;  \nllabel{alg:FO}

  \BlankLine

  $\Psi \leftarrow {\sf OneBlockQuantifierElimination}(\exists, \Phi, [\X, \Y], [Z, \zeta, \eta, \eps])$ \; \nllabel{alg:QE}

  \BlankLine
  \CommentSty{/* {$\Psi = [\Psi_1, \dots, \Psi_I]$} */}

  \BlankLine

  $\widetilde \Psi \leftarrow [{\sf UnivariateLimit}(\Psi_1,\eps), \dots {\sf Limit}(\Psi_I, \eps)]$ \; \nllabel{alg:limit}
  
  \BlankLine
  \For{$1 \leq k \leq I$ }{   
    \BlankLine
    \lIf{ {\sf IsRealizable}$(\widetilde \Psi_k \wedge (Z >0))$ }{ 
      \RET {\sf True} \; \nllabel{alg:realizable}
    }
     
  }
  \RET {\sf False} \;
\end{algorithm2e}
}

For the correctness proof of \func{DisjointPolar} we will need the following lemma.
\begin{lemma}
  \label{lem:exists-x}
  Assume that $f$ is non-negative over $\R^n$ and satisfies $\func{N}$.
  If $\x\in \left (\limeps \crit(\pi_{i}, V_\eps)\right ) \inter \R^n$, 
  then there exists $\x_\eps \in \crit(\pi_{i},V_\eps) \inter \Re^n$ 
  such that $\limeps\x_\eps=\x$. 
\end{lemma}
\begin{proof}
  Consider the  polynomial
  $H = f + \sum_{k=i}^{n} (\frac{\partial f} {\partial X_k} )^2 $
  and notice that $H(\x) = 0$.
  Then, apply Proposition~\ref{prop:comp-deform} to the semi-algebraically connected component of
  the real solution set of $H = 0$ containing $\x$.
  Thus, there exists $\x_{\eps} \in \Re^n$ such that
  $H(\x_{\eps}) - \eps = 0$
  and $\limeps \x_{\eps} = \x$.
  This implies that $\x_{\eps} \in \crit(\pi_{i-1}, V_{\eps}) \inter \Re^n$.
\end{proof}

We can now prove Lemma~\ref{lem:disjoint-polar}.
  
{\bf Correctness.}  By assumption, $f$ is non-negative over $\RRz^n$
and satisfies $\sf N$. Thus  $\crit(\pi_1, V_\eps)$ is
finite. Moreover, all the assumptions of Lemma~\ref{lemma:limitsoffiniteset}
are satisfied that implies correctness when $i=1$.

\vspace{10pt} 
In what follows, we assume that $i\geq 2$.

Note that the systems $P_1$ and $P_2$ at Steps \ref{alg:F-def} and \ref{alg:G-def} define respectively the sets 
$$
\crit(\pi_i, V_\eps)\inter U_{\eta, \eps}\inter B_\eps
\quad \text{ and }\quad 
\crit(\pi_{i-1}, V_\eps)\inter U_{\eta, \eps}\inter B_\eps
\enspace .
$$

At Step \ref{alg:emptytest}, the call to \func{IsEmpty} returns {\sf True} if
and only if $\crit(\pi_{i-1}, V_\eps)\inter U_{\eta,\eps}\inter B_\eps$ is
empty. We claim that if this latter set is empty then $\left (\limeps
  \crit(\pi_{i-1}, V_\eps)\right )\inter U$ is empty. Since by
assumption $\left (\limeps \crit(\pi_{i}, V_\eps)\right )\inter U$ is
not empty, \func{DisjointPolar} runs correctly by returning {\sf True}
at Step \ref{alg:emptytest}. Now we prove our claim. 

Assume by contradiction that there exists 
$\x\in \left (\limeps \crit(\pi_{i-1}, V_\eps)\right )\inter U$.
 Then, by assumption $\x\in
B$ and $g_1(\x)>0, \ldots, g_s(\x)>0$. We deduce that for $1\leq k
\leq s$, we have $g_k(\x) > \eta$. 
The latter inequality is strict because $\x \in \RRz^n$ and the coefficients of $g_k$
lies in $\RRzh$.

Since $f$ is non-negative over $\RRz^n$, 
there exists $\x_\eps\in \crit(\pi_{i-1},V_\eps)\cap\Re^n$ 
such that $\limeps\x_\eps=\x$, by Lemma~\ref{lem:exists-x}. 
As a consequence, for
$1\leq k \leq s$ we have $g_k(\x_\eps)\geq \eta$ (otherwise, using
$\limeps$, this would contradict $\x\in U_\eta$). As a consequence, we
have $\x_\eps \in \crit(\pi_{i-1}, V_\eps)\inter U_{\eta,
  \eps}$. Finally, since $\x\in \left (\limeps \crit(\pi_{i-1},
  V_\eps)\right )\inter U$, it is at positive distance to the boundary
of $B$. Since $\limeps\x_\eps$ is $\x$,
$\x_{\eps}$ and $\x$ are infinitesimally close
and $\x_\eps$ lies in $B_\eps$ at positive distance from its boundary
w.r.t $\eps$. We conclude that there exists $\x_\eps\in
\crit(\pi_{i-1}, V_\eps)\inter U_{\eta,\eps}\inter B_\eps$ which is a
contradiction.

\vspace{15pt}

In the rest of the proof we can assume now that $\crit(\pi_{i-1},
V_\eps)\inter U_\eps\inter B_\eps$ is not empty.

In the discussion preceding the correctness proof we considered the semi-algebraic set 
\begin{displaymath}
  \begin{aligned}
    A_{\eps} = \{ (\x, \y, z) \in \Re^{2n+1} &\quad |&& 
    \x\in \left (\crit(\pi_{i}, V_\eps) - \crit(\pi_{i-1}, V_\eps) \right )\inter U_{\eta,\eps}\inter B_\eps,  \\
    &&&\y\in \crit(\pi_{i-1}, V_\eps)\inter U_{\eta,\eps}\inter B_\eps, \\
    &&&\, ||x-y||=z > 0\}
     \enspace .
  \end{aligned}
\end{displaymath}
We can also describe $A_{\eps}$ as 
\[
A_{\eps} = \{ (\x, \y, z) \in \Re^{2n+1} | 
P_1(\x) \wedge P_2(\y) 
\wedge  \frac{\partial f (\X)}{\partial X_{i}} \not= 0 
\wedge  \| \x - \y \| > z^2 \wedge z>0\} 
\enspace .
\]
We refer to Steps \ref{alg:F-def} and \ref{alg:G-def} for the
definitions of $P_1$ and $P_2$ and the polynomials that they
involve.

If we replace the three infinitesimals $\eps$, $\eta$, and $\zeta$ with three new
variables $e$, $h$, and $t$ respectively, in the formula $\Phi$ defined at Step
\ref{alg:QE}, then we get polynomials $P_1(\X, e,h, t)$, 
$\frac{\partial f(\X, e, h,t)}{\partial X_{i}}$, and $P_2(\Y, e,h, t)$.
In this way we define the following semi-algebraic set
\[
A = \{ (\x, \y, z, e, h, t) \in \RR^{2n+4} \,|\, 
P_1(\x, e, h, t) \wedge  P_2(\y, e,h, t)  \wedge
\frac{\partial f (\X,e,h, t)}{\partial X_{i}} \not= 0  \wedge
\| \x - \y \| > z^2 \wedge z>0\}
\enspace .
\]

By the definition of an extension ($\ext$), see
Section~\ref{sec:prelim}, of a semi-algebraic set we get
\begin{equation}
  \label{eq:pi-equiv}
  \pi_{\x, \y, z}\left( \ext(A, \Re) \inter \{ e=\eps, h=\eta, t=\zeta \} \right) 
  = A_{\eps} \subset \Re^n
  \enspace,
\end{equation}
where $\pi_{\x, \y, \z}: (\x, \y, z, e,h, t) \mapsto (\x, \y, z)$. We
deduce that 
\[
\pi_{z}( \pi_{\x,\y\z}( \Ext(A, \Re) \inter \{e=\eps, h=\eta, t=\zeta \}))  = \pi_{z}(A_{\eps})
\enspace ,
\]
where $\pi_z: (\x, \y, z) \mapsto z$.

The application of \func{OneBlockQuantifierElimination} at $\Phi$
(Step~\ref{alg:QE}) actually computes a semi-algebraic formula
defining $\pi_{z,e,h, t}(A)$.  By (\ref{eq:pi-equiv}) and the above
discussion, one can conclude that substituting $e$, $h$ and $t$ by
$\eps$, $\eta$ and $\zeta$ in this formula provides a semi-algebraic
formula defining $\pi_z(A_\eps)$.

\vspace{15pt}

Finally, we have to prove that indeed \func{IsRealizable} returns
the correct answer. This means that the following formula is true:
\begin{displaymath}
  \exists z \in \R, z > 0 \wedge z \in \limeps \pi_z(A_\eps) 
  \Leftrightarrow 
  \left (\limeps \crit( \pi_i, V_{\eps})\right ) \inter U \not= \left (\limeps \crit( \pi_{i-1}, V_{\eps})\right ) \inter U
  \enspace ,
\end{displaymath}
for some given value of $i \in \{2, \dots, n\}$.

\noindent  {\bf The if statement.} 
If such a $z$
exists, then there are distinct points $\x_\eps, \y_\eps$ in $\Re^n$
and $z_{\eps} \in \Re$, such that $(\x_{\eps}, \y_{\eps}, z) \in A_\eps$, that is
\begin{eqnarray*}
    &\x_\eps\in (\crit( \pi_i, V_{\eps})-\crit( \pi_{i-1}, V_{\eps})) \inter U_{\eta, \eps}\inter B_\eps,
    \quad 
    \y_\eps \in \crit( \pi_{i-1}, V_{\eps}) \inter U_{\eta, \eps}\inter B_\eps  \enspace,\\
    &\text{ and } \limeps ||\x_\eps-\y_\eps||=z^2>0 
\enspace.
\end{eqnarray*}
This means that $\x_\eps$ and $\y_\eps$ are not infinitesimally 
close, w.r.t. $\eps$.  Notice that,
by definition of $A_\eps$, $\x_\eps$ and $\y_\eps$ both lie in
$B_\eps$ that is bounded over $\R$.
Consequently, they are bounded
over $\R$ and their limits as $\eps \rightarrow 0$ exist. 
The limits as $\eps \rightarrow 0$
are different and thus 
$$\limeps \left (\crit( \pi_i, V_{\eps}) \inter  U_{\eta, \eps}\inter B_\eps\right ) 
\not= 
\limeps \left (\crit( \pi_{i-1}, V_{\eps}) \inter U_{\eta,\eps}\inter B_\eps\right ) ,
$$ 
because $\x_{\eps}$ and $\y_{\eps}$ are not infinitesimally close.
Notice that $\limeps \x_\eps$ lies in $\limeps\crit(\pi_i, V_\eps)$ and
$U_\eta\subset \ext(U, \R)$.
Similarly $\limeps \y_\eps$ lies in $\limeps\crit(\pi_{i-1}, V_\eps)$ and
 $U_\eta\subset \ext(U, \R)$.
Combining these inclusions  with the fact that $\x_\eps$ and $\y_\eps$ are not infinitesimally close w.r.t $\eps$,
we conclude that 
$$   
\left (\limeps \crit( \pi_i, V_{\eps})\right ) \inter U_{\eta} \not= 
\left (\limeps \crit( \pi_{i-1}, V_{\eps})\right ) \inter U_{\eta} 
\enspace .
$$

\noindent
{\bf The only if statement.} 
Assume that
\[
\left (\limeps \crit( \pi_i, V_{\eps}) \right ) \inter U \not=
\left (\limeps \crit( \pi_{i-1}, V_{\eps})\right ) \inter U \enspace.
\]

We have to prove that there exists $(\x_{\eps}, \y_{\eps}, z_{\eps}) \in A_{\eps}$
and that $\limeps\| \x_{\eps} - \y_{\eps} \| = \limeps z_\eps > 0$.

First, we prove the existence of $\x_{\eps}$. 
Take 
$\x\in \left ((\limeps \crit( \pi_i, V_{\eps}))  \inter U 
\setminus (\limeps \crit( \pi_{i-1}, V_{\eps})) \inter U\right )$. 
By noticing that $\x\in U_\eta$, we deduce that 
$$
\x\in \left ((\limeps \crit( \pi_i, V_{\eps}))\inter U_\eta 
  \setminus 
  (\limeps \crit( \pi_{i-1}, V_{\eps})) \inter U_\eta\right )
\enspace .
$$

We claim that this implies that there exists 
$\x_{\eps} \in \crit(\pi_i, V_{\eps}) \inter U_{\eta,\eps} $ such that
$$\x = \limeps \x_{\eps} \in  \left (\limeps \crit(\pi_i, V_{\eps})\right ) \inter U_\eta.$$
Indeed, since $\x\in \limeps \crit( \pi_i, V_{\eps})$ and $f$ is
non-negative over $\RRz^n$, 
we deduce that there exists $\x_\eps\in \crit(\pi_i, V_\eps)\cap\Re^n$
with $\limeps\x_\eps=\x$ (Lemma~\ref{lem:exists-x}).

Moreover,  $\x\in U_\eta$.
But $\x\in \RRz$ and the coefficients of the
$g_i$'s lie in $\RRz$, thus  $g_1(\x)> \eta, \ldots, g_s(\x)> \eta$.
If there exists a $k$,  $1\leq k \leq s$,
such that $g_k(\x_\eps)\leq \eta$, then this would imply that
$g_k(\limeps \x_\eps)=g_k(\x)\leq \eta$.
This contradicts the fact that $\x \in U_{\eta}$ and so $g_k(\x)>\eta$
for all $k$.
We conclude that $\x_\eps\in U_{\eta, \eps}$
and so  $\x_\eps  \in \crit(\pi_i, V_{\eps}) \inter  U_{\eta, \eps}$
as we claimed.

We also deduce that 
$\x_{\eps} \not\in \crit( \pi_{i-1}, V_{\eps}) \inter U_{\eta, \eps}$. 
Since in this case 
\[
\x=\limeps \x_{\eps} \in \left (\limeps \crit( \pi_{i-1}, V_{\eps})\right ) \inter U
\enspace,
\]
which is a contradiction.  
Thus, there exists 
$\x_\eps\in (\crit( \pi_i, V_{\eps})-\crit( \pi_{i-1}, V_{\eps})) \inter U_{\eta, \eps}\inter B_\eps$.
The inclusion  $\x_{\eps} \in B_{\eps}$ is a direct consequence of the fact 
that $B_{\eps}$ strictly contains the semi-algebraic set under consideration.

Now we prove the existence of $\y_{\eps}$.
Recall that 
$\x$ is obtained as the limit of $\x_{\eps}$ that it is not infinitesimally close, w.r.t $\eps$,
to any point of the set $\crit( \pi_{i-1}, V_{\eps})$.
By the assumption that 
$\crit(\pi_{i-1}, V_{\eps}) \inter U_{\eta, \eps}$ is not empty, 
there  exists $\y_{\eps} \in \crit( \pi_{i-1}, V_{\eps}) \inter U_{\eta,\eps}$ 
that is not infinitesimally close to $\x_{\eps}$,
w.r.t. $\eps$.

As there exist $\x_{\eps}$ and $\y_{\eps}$ that are not infinitesimally close w.r.t. $\eps$,
there exists $z_\eps \in \Re$ such that
$z_\eps >0$ and $\| \x_{\eps} - \y_{\eps}\|^2 > z_\eps^2$ and $\limeps z_\eps>0$.


\medskip \textbf{Complexity analysis.}  
When \func{DisjointPolar} is called with $i=1$ its run time is the one
of \func{LimitsOfOfCriticalPoints}. By
Lemma~\ref{lemma:limitsoffiniteset}, Step \ref{alg:ifieq1} costs $(s
\, D)^{O(n)}$ arithmetic operations.

Now we assume that $i\geq 2$. Steps \ref{alg:F-def}-\ref{alg:G-def}
and Step \ref{alg:FO} are symbolic manipulations which produce Boolean
combinations of $O(s)$ polynomials of total degree $2D$ involving
$O(n)$ variables.

By Lemma \ref{lem:subroutine-IsEmpty}, Step \ref{alg:emptytest} costs
at most $(s\, D)^{O(n)}$ arithmetic operations.

The complexity of one-block quantifier elimination at Step
\ref{alg:QE} is $(s\,D)^{O(n)}$ (Theorem~\ref{th:qe}).

The output of the quantifier elimination procedure consists of $(s\,
D)^{O(n)}$ conjunctions of polynomials in $\ZZ[\zeta, \eta, Z][\eps]$ of
degree at most $D^{O(n)}$.  Each conjunction involves at most $(s \,
D)^{O(n)}$ polynomials.

After we apply \func{UnivariateLimit} (Step~\ref{alg:limit}) we check if there
is a realizable sign condition.  According to Lemma
\ref{lem:subroutinelimit} this costs $(s\, D)^{O(n)}$ arithmetic
operations.

We call \func{IsRealizable} at most $(s\, D)^{O(n)}$ times (which is
the number of conjunctions and thus the cardinality of $\widetilde
\Psi$).  Each call involves $(s \, D)^{O(n)}$ polynomials in
$\ZZ[\zeta,\eta,Z]$ of degree at most $D^{O(n)}$ (Lemma
\ref{lemma:IsRealizable}).  The arithmetic cost is
$(s\,D^{O(n)})^{O(1)}= (s\,D)^{O(n)}$ (Sec.~\ref{sec:realizable})
which is also the cost for the whole for-loop and the algorithm.

Statements on bit complexity when $\sf Z=\Z$ are straightforward
applying {\em mutatis mutandis} the same reasoning as above and
using bit complexity results given in Theorem \ref{th:qe}, and
Lemmata \ref{lem:subroutinelimit}, \ref{lemma:limitsoffiniteset},
and  \ref{lemma:IsRealizable}.
\hfill \qed

\subsection{Proofs of subroutines}
\printproofs


\subsection*{Acknowledgments}

Both authors are partially supported by the EXACTA grant of the
National Science Foundation of China (NSFC 60911130369) and the French
National Research Agency (ANR-09-BLAN-0371-01), GeoLMI (ANR 2011 BS03
011 06), HPAC (ANR ANR-11-BS02-013).  Safey El Din is partially
supported by the Institut Universitaire de France.  Elias Tsigaridas
is partially supported by an FP7 Marie Curie Career Integration Grant.

{   
\bibliographystyle{abbrv}  
\bibliography{theo1}
}

\end{document}